\documentclass[lettersize,journal]{IEEEtran}
\usepackage{amsmath,amsfonts}
\usepackage{amsthm}
\usepackage{algorithmic}
\usepackage{algorithm}
\usepackage{array}
\usepackage[caption=false,font=normalsize,labelfont=sf,textfont=sf]{subfig}
\usepackage{textcomp}
\usepackage{stfloats}
\usepackage{url}
\usepackage{verbatim}
\usepackage{graphicx}
\usepackage{cite}
\usepackage{hyperref}
\usepackage{booktabs}
\usepackage{nicefrac}
\usepackage{microtype}
\usepackage{xcolor}
\usepackage{subfig}
\usepackage{amssymb}
\usepackage{mathtools}
\usepackage{enumitem}
\usepackage{cleveref}
\theoremstyle{plain}
\theoremstyle{definition}
\newtheorem{theorem}{Theorem}
\newtheorem{lemma}{Lemma}
\newtheorem{corollary}{Corollary}
\newtheorem{proposition}{Proposition}
\newtheorem{definition}{Definition}
\newtheorem{example}{Example}
\theoremstyle{remark}
\newtheorem{remark}{Remark}
\begin{document}

\title{Multi-group Bayesian Games}

\author{Hongxing Yuan, Xuan Zhang, Chunyu Wei*, and Yushun Fan*,~\IEEEmembership{Member,~IEEE}
\thanks{}
\thanks{Hongxing Yuan, Xuan Zhang are with the Department of Automation, Tsinghua University, and Beijing National Research Center for
Information Science and Technology, Beijing 100084, China.}
\thanks{*Chunyu Wei is with the School of Information, Renmin University of China, Beijing 100872, China.}
\thanks{*Yushun Fan is with the Department of Automation, Tsinghua University, and Beijing National Research Center for
Information Science and Technology, Beijing 100084, China (e-mail: 
fanyus@tsinghua.edu.cn). }
\thanks{This work was supported by the National Nature Science Foundation of China under Grant 62173199.}}

\markboth{Journal of \LaTeX\ Class Files,~Vol.~14, No.~8, August~2021}
{Shell \MakeLowercase{\textit{et al.}}: A Sample Article Using IEEEtran.cls for IEEE Journals}

\IEEEpubid{0000--0000/00\$00.00~\copyright~2021 IEEE}

\maketitle

\begin{abstract}
This paper presents a model of multi-group Bayesian games (MBGs) to describe the group behavior in Bayesian games, and gives methods to find (strongly) multi-group Bayesian Nash equilibria (MBNE) of this model with a proposed transformation. MBNE represent the optimal strategy \textit{profiles} under the situation where players within a group play a cooperative game, while strongly MBNE characterize the optimal strategy \textit{profiles} under the situation where players within a group play a noncooperative game. Firstly, we propose a model of MBGs and give a transformation to convert any MBG into a multi-group ex-ante agent game (MEAG) which is a normal-form game. Secondly, we give a sufficient and necessary condition for a MBG's MEAG to be (strongly) potential. If it is (strongly) potential, all its (strongly) Nash equilibria can be found, and then all (strongly) MBNE of the MBG can be obtained by leveraging the transformation's good properties. Finally, we provide algorithms for finding (strongly) MBNE of a MBG whose MEAG is (strongly) potential and use an illustrative example to verify the correctness of our results.
\end{abstract}


\section{Introduction}
\label{Introduction}

\IEEEPARstart{C}{lassical} game theory was first proposed by mathematicians \cite{Neumann1944}. A normal-form game \cite{Frihauf2012} contains a player set, an action \textit{profile} set and a payoff function. It is a commonly used model of games with complete information which can be categorized into two types: noncooperative games and cooperative games. Since each player may not have complete knowledge of the action sets and payoff functions of other players in practice, scholars have proposed two models of games with incomplete information: Harsanyi model \cite{Harsanyi1} and Aumann model \cite{Aumann1976}.

In Harsanyi games \cite{Harsanyi2} which are also called Bayesian games, factors type and belief are added to describe the uncertainty. Bayesian games can be used to describe and solve scenarios with incomplete information, so they have received more and more attention in the fields of machine learning and artificial intelligence, such as auctions \cite{bidding1,Elyakime1994}, routing problems \cite{routing1,routing2} , dynamic pricing \cite{DynPri1,DynPri2}, security and defense \cite{Komal2024,defense1}, multi-agent reinforcement learning \cite{multiagent2024,multirei1,multirei2}, multi-agent communication \cite{multicom1,multicom2}, and so on. Aumann model describes the concept of common knowledge through the mathematical definition of knowledge in incomplete information. It can be proved that Harsanyi model and Aumann model are equivalent under certain conditions \cite{Aumann1976}.

Nash equilibrium is a pivotal concept in game theory, where no player wants to adjust his action unilaterally if the others' actions remain unchanged. Thus, it's an optimal operating point in practice. A Bayesian game is with uncertainty, so each player's action selection needs to be based on his type. Therefore, a Bayesian Nash equilibrium is not a fixed action \textit{profile}, but a strategy, which is a mapping from the type \textit{profile} set to the action \textit{profile} set \cite{Harsanyi1,Gibbons}. It's difficult to directly find all Bayesian Nash equilibria on a mapping space. Therefore, scholars proposed several transformations \cite{Harsanyi2,Harsanyi3,Wu} to convert a Bayesian game into a normal-form game which is with complete information. In this way, Bayesian Nash equilibria can be found by finding the normal-form game's Nash equilibria.

Group behavior is so common in real life that it has been introduced into many fields \cite{Glanzer,Perc,Alvarez-Rodriguez}. In practical scenarios, players are typically embedded in groups rather than acting as isolated individuals. Thus, the group behavior in game theory also holds substantial practical significance. For games with complete information, scholars put forward a concept of multi-cluster game and gave a method to find its Nash equilibria \cite{Ye2018,Meng,Zeng}. In multi-cluster games, the goal of each group is to maximize its payoff, which is the average payoff of all players within the group. Each group assigns a representative player to interact with representative players from other groups through the network.

For games with incomplete information (i.e., Bayesian games), existing studies on group behavior primarily focus on whether players choose to cooperate \cite{ichiishi1992bayesian}, how to design cooperative payoff allocation mechanisms to prevent players from misreporting their types\cite{safronov2018coalition}. As far as we know, there is no relevant research addressing scenarios where the group structure and cooperative payoff allocation mechanism are preassigned according to actual scenarios, and players aren't required to choose whether to misreport. However, in real life, these scenarios are very common.

For example, due to the limited communication ability, all agents are divided into groups. Within the groups, agents play games with complete information, while between groups, agents play games with incomplete information \cite{multiUAV1,multiUAV2}. Depending on the task, the relationships between agents within a group can be divided into two situations \cite{TwoSitua1,TwoSitua2}.
\IEEEpubidadjcol
If agents within a group are working together to complete a task, then they play a cooperative game with complete information. If agents within a group have different tasks, then they play a noncooperative game with complete information. 

Similar dynamics emerge in social groupings: individuals typically possess full information about peers within their social circles but limited knowledge of outsiders \cite{auction2024,auction2}. Also, people within a social circle may choose strategies together and share their payoffs equally to reduce risks. In the other situation, they may just be familiar with each other, but choose strategies independently to maximize their own payoffs.

To describe and investigate the above scenarios, we propose and analyze a novel game termed multi-group Bayesian game (MBG). As two fundamental challenges in analyzing any game are: the verification of a game's potentiality and the methods to find Nash equilibria, we will elaborate on the potentiality of a game and introduce a mathematical tool for finding Nash equilibria of a game, namely the semi-tensor product (STP) in the following. 

If a game is potential, then it has Nash equilibria which can be obtained by finding any potential function of the game \cite{Rosenthal,Heumen}. In addition, this game will converge to pure Nash equilibria in the process of dynamic evolution. Considering these good properties of potential games, they are favored in practical applications \cite{potential}. And the STP of matrices has been verified to be a useful tool in studying potential games.

The semi-tensor product (STP) of matrices was first proposed by Professor Cheng \cite{cheng2011}. In contrast to the traditional matrix product, it has no limitation on the dimensions of the two matrices, so it can be carried out between any two matrices. It has been widely used to study problems in many fields such as Boolean networks\cite{lu2025link,liu2025analysis,feng2023original,Li2022reachability,cheng2021self}, mix-valued logical networks\cite{li2021,Lin2023,lu2020event}, discrete event systems\cite{zhang2022initial,zhou2024three}. In particular, many interesting results have been obtained in the study of games by using the STP of matrices \cite{Zhu2023group,Yuan2022,feng2025iterative,wang2024stability,Wu}. 

In the study of games' potentiality, with the help of the STP, the problem of verifying whether a game is a potential can be transformed into a problem of verifying whether a linear system which is called the potential equation has a solution, and if there exists a solution, the potential function can be obtained directly from the solution. Then we can easily get all Nash equilibria\cite{Wu}. Therefore, the computational complexity of verifying the game potentiality was reduced by introducing the STP \cite{cheng2014}. Furthermore, \cite{Liu2016} improved this method and proved that the computational complexity of the improved method has reached the minimum.

Motivated by the above, we propose the concepts of the multi-group Bayesian game (MBG) and its multi-group ex-ante agent game (MEAG) which is a normal game. Then we give the properties of this multi-group ex-ante agent transformation. There is a one-to-one correspondence between each (strongly) Nash equilibrium of the MBG and its MEAG, and this transformation preserves (strongly) potentiality. Therefore, as long as the MBG's MEAG is (strongly) potential, no matter whether the MBG is (strongly) potential or not, we can obtain the MBG's (strongly) multi-group Bayesian Nash equilibria by seeking for (strongly) Nash equilibria of its MEAG. 

In order to obtain the MEAG's (strongly) Nash equilibria, we give the algebraic form and (strongly) potential equation of the MEAG. If there exists a solution to the (strongly) potential equation, then the MEAG is (strongly) potential, and its (strongly) potential functions can be obtained from the solutions to the (strongly) potential equation. After finding out the MEAG's all (strongly) Nash equilibria by using an arbitrary (strongly) potential function, we can obtain all (strongly) multi-group Bayesian Nash equilibria of its MBG based on the one-to-one correspondence described above.

The main innovations of this paper are:
\begin{itemize}[leftmargin=*, topsep=0pt, itemsep=0pt, parsep=0pt]
  \item propose a novel game called the MBG to describe the group behavior in a Bayesian game and provide two types of definitions for Nash equilibria and potentiality, where we use ``strongly'' to distinguish cooperative and non-cooperative behaviors in groups. 
  \item give a transformation which converts a MBG to a MEAG. Then construct the algebraic form and (strongly) potential equation of a MBG's MEAG via the STP of matrices. Through the (strongly) potential equation, the (strongly) potentiality of the MEAG can be verified. And if the MEAG is (strongly) potential, then all its (strongly) Nash equilibria can be obtained.
  \item provide the properties of the above transformation, based on which we can obtain (strongly) multi-group Bayesian Nash equilibria of MBGs through finding (strongly) Nash equilibria of the (strongly) multi-group ex-ante agent potential games.
\end{itemize}

The arrangement for the rest of this paper is as follows. Some essential notations and concepts are given in \Cref{Pre}. \Cref{MBG} proposes the definitions and concepts of the MBG and its MEAG, then introduces and proves the properties of this transformation. In \Cref{Ver}, the algebraic form of a MBG's MEAG and methods for finding (strongly) MBNE are proposed, which is followed by a brief conclusion in \Cref{sec:conclusion}.

\section{Preliminaries\label{Pre}}

In this section, we give some essential notations and concepts of the semi-tensor product (STP) and Bayesian games. 

\subsection{Notations}
\begin{itemize}[leftmargin=*]
  \item $\mathbb{R}$ denotes the set of real numbers.
  \item $\ltimes$ represents the STP, which generalizes the ordinary matrix product. And for convenience, it may be omitted hereafter.
  \item $\otimes$ is the Kronecker product.
  \item $\mathcal{D}_r:=\{1, 2, \cdots, r\}$, $\mathbf{1}_{r}:=[1,1, \cdots ,1]^{T}$.
  \item $\mathbb{R}_{l \times k}$ denotes the set of all $l \times k$ dimension real matrices.
  \item $I_{r}$ represents the $r$ dimension identify matrix.
  \item For $L \in \mathbb{R}_{m \times n}$, $\operatorname{Col}_{l}(L)$ denotes the $l$-th column of matrix $L$, and $\operatorname{Col}(L)$ indicates $\{\operatorname{Col}_{l}(L) : 1 \leq l \leq n\}$.
  \item $\delta_{r}^{l}:=\operatorname{Col}_{l}(I_{r})$, $\Delta_{r}:=\{\delta_{r}^{l} : 1 \leq l \leq r\}$.
  \item $L \in \mathbb{R}_{m \times n}$ is termed as a logical matrix, if $\operatorname{Col}(L) \subseteq \Delta_{m}$. 
  \item $\mathcal{L}_{m \times n}$ expresses the set of all $m \times n$ dimension logical matrices.
  \item $\mid G \mid$ stands for the number of elements in set $G$.
\end{itemize}

\subsection{Semi-tensor product of matrices\label{stp}}

We give some basic concepts about the STP and its excellent properties used in the derivation of this paper. One can refer to \cite{cheng2011} for more details.

\begin{definition}\label{def1}\cite{cheng2011}
The semi-tensor product of matrices $A \in \mathbb{R}_{a\times b}$ and $B\in \mathbb{R}_{ c\times d}$ is defined as
\begin{equation}
A\ltimes B=(A\otimes I_{\frac{x}{b}})(B\otimes I_{\frac{x}{c}}),
\end{equation}
where $x=lcm(b,c)$ is the least common multiple of $b$ and $c$.
\end{definition}

\begin{lemma}\cite{cheng2011}\label{lem1}
\begin{enumerate}[leftmargin=*, topsep=0pt, itemsep=0pt, parsep=0pt]
  \item[(i)] For $Q \in \mathbb{R}_{q \times 1}$ and $C \in \mathbb{R}_{a \times b}$, we have
  $$
  Q \ltimes C=(I_{q} \otimes C) \ltimes Q.
  $$
  \item[(ii)] For $H \in \mathbb{R}_{h \times 1}$ and $Q \in \mathbb{R}_{q \times 1}$, we have
  $$
  Q \ltimes H=W_{[h, q]} \ltimes H \ltimes Q,
  $$
  where $W_{[h, q]} \in \mathbb{R}_{hq \times hq}$ denotes the swap matrix which satisfies $\operatorname{Col}_{(i-1) q+j}(W_{[h, q]})=\delta_{hq}^{i+(j-1) h}, i=$ $1, \cdots, h, j=1, \cdots, q$.
  \item[(iii)] For $Q \in \Delta_{q}$, we have
  $$
  Q^{2}=O_{q} Q,
  $$
  where $O_{q}:=\delta_{q^{2}}[1,q+2,2q+3 ,\cdots,q^{2}] \in \mathcal{L} _{q^{2} \times q}$ denotes the order reducing matrix.
  \item[(iv)] For $H \in \Delta_{h}, Q \in \Delta_{q}, X \in \Delta_{x}$, we have
  \begin{equation*}
  M_{[h, q, x]} H Q X=Q,\label{M_nlm}
  \end{equation*}
  where $M_{[h, q, x]}=\mathbf{1}_{h}^{T} \otimes I_{q} \otimes \mathbf{1}_{x}^{T}$.
  \item[(v)] For $A_i \in \Delta_{b_i}$, $M_i \in \mathbb{R}_{a_i \times b_i}$, $i \in [1;r]$, we have
  \begin{equation*}
  {\ltimes}_{i=1}^{r}(M_i  A_i)=( {\otimes}_{i=1}^{r}M_i) ({\ltimes}_{i=1}^{r} A_i).\label{ju_xiang}
  \end{equation*}
  \item[(vi)] For $M \in \mathbb{R}_{a \times b}$, we have
  \begin{equation*}
  W_{[a,c]} \ltimes M \ltimes W_{[c,b]}=I_c \otimes M.
  \end{equation*}
  Therefore we get $W^{T}_{[a,b]}=W_{[b,a]}$.
\end{enumerate}
\end{lemma}

\begin{lemma}\cite{cheng2011}\label{lem2} For a pseudo-logical function $g: \prod_{i=1}^n \mathcal{D}_{k_i}\\  \rightarrow \mathbb{R}$, identify $\mathcal{D}_{k_i} \sim \Delta_{k_i}$, then $g$'s algebraic form can be expressed as
\begin{equation}\label{alg}
g(x_{1}, \cdots, x_{n})= L_{g} \ltimes_{i=1}^{n} x_{i} \in \mathbb{R},
\end{equation}
where $L_{g} \in \mathcal{L} _{1 \times \prod_{i=1}^n k_i}$ is the structure matrix of $g$, $ x_{i} \in \Delta_{k_i}, i=1, \cdots, n.$
\end{lemma}

A straight corollary of \Cref{lem2} is as follows.

\begin{corollary}\label{cor1} For pseudo-logical functions $f,g: \prod_{i=1}^n \Delta_{k_i}  \rightarrow \mathbb{R}$ whose structure vectors are $L_{f}$ and $L_{g}$ respectively, we have
\begin{equation}
\sum_{x \in \Delta_{k}} f(x)g(x)=L_{f} L_{g}^{T},
\end{equation}
where $x=\ltimes_{i=1}^{n} x_{i}$, $k=\prod_{i=1}^n k_i.$
\end{corollary}

\subsection{Bayesian games\label{Bayesian}}

Games can be divided into two categories based on whether they are with complete information or not. In a normal-form game, which is a commonly used model of games with complete information, each player completely knows all key information of other players. However, in a Bayesian game, which is a typical game with incomplete information, at least one player can't know certain key information of other players, such as payoff functions.

\begin{definition}\cite{Gibbons} \rmfamily\label{def2}
A normal-form game is denoted by a tuple $\mathcal{G}=( G,\mathcal{A},c)$, where
\begin{itemize}[leftmargin=*, topsep=0pt, itemsep=0pt, parsep=0pt]
  \item $G=\{1, \cdots, m\}$ indicates a player set.
  \item $\mathcal{A}={\prod}_{i=1}^m \mathcal{A}_i$ is an action \textit{profile} set, where $\mathcal{A}_i$ stands for the action set of player $i$.
  \item $c=\{c_1,\cdots,c_m\}$ is a payoff function set, where $c_i: \mathcal{A} \rightarrow \mathbb{R}$ denotes player $i$'s payoff function.
\end{itemize}
\end{definition}

\begin{definition}\cite{Gibbons} \rmfamily\label{def3}
A Bayesian game is denoted by a tuple $\mathcal{G}=( G,\mathcal{T},\mathcal{A},p,c)$, where
\begin{itemize}[leftmargin=*, topsep=0pt, itemsep=0pt, parsep=0pt, partopsep=0pt]
  \item $G=\{1, \cdots, m\}$ represents a player set.
  \item $\mathcal{T}={\prod}_{i=1}^m \mathcal{T}_i$ is termed as a type \textit{profile} set, where $\mathcal{T}_i$ denotes the type set of player $i$.
  \item $\mathcal{A}={\prod}_{i=1}^m \mathcal{A}_i$ expresses an action \textit{profile} set, where $\mathcal{A}_i$ denotes the action set of player $i$.
  \item \begin{equation}
  p\left(t_{-i} \mid t_i\right)=\frac{p\left(t_{-i},t_i\right)}{p\left(t_i\right)}=\frac{p\left(t_{-i}, t_i\right)}{\sum_{t_{-i}^{\prime} \in \mathcal{T}_{-i}} p\left(t_{-i}^{\prime}, t_i\right)}
  \end{equation}
  stands for a belief, where $\mathcal{T}_{-i}={\prod}_{l \in G, l\neq i} \mathcal{T}_{l}$. $p: \mathcal{T} \rightarrow[0,1]$ denotes a probability distribution over $\mathcal{T}$, which is a common knowledge. And $t_i$ is player $i$'s private knowledge.
  \item $c=\{c_1,\cdots,c_m\}$ indicates a payoff function set, where $c_i:\mathcal{T} \times \mathcal{A} \rightarrow \mathbb{R}$ denotes player $i$'s payoff function.
\end{itemize}
\end{definition}

In a Bayesian game, each player's payoff function is determined by his type and is no longer fixed. Each player can't know others' types, but can know the probability distribution for all players' types, so one can make a probability estimation of other players’ types, which is regarded as his belief. Below, we will give a simple example detailing how to model auction problems using Bayesian games to help understand the additional elements type and belief in a Bayesian game compared to a normal game.

\begin{example}\label{e1}
An auction whose auction rule is the first-price sealed-bid auction \cite{Elyakime1994} can be modeled as a Bayesian game and represented by a tuple $\mathcal{G}=( G,\mathcal{T},\mathcal{A},p,c)$, where
\begin{itemize}[leftmargin=*, topsep=0pt, itemsep=0pt, parsep=0pt]
  \item $i \in G=\{1, \cdots, 3\}$ represents a bidder.
  \item $\mathcal{T}_1=\{t_{11}=100,t_{12}=110\}$ denotes bidder $1$'s all possible evaluations of the item, that is to say, the auction item is worth $100$ or $110$ for bidder $1$. Similarly, the type sets of the other two bidders are $\mathcal{T}_2=\{t_{21}=108,t_{22}=93\}$ and $\mathcal{T}_3=\{t_{31}=78,t_{32}=95\}$. The real evaluation \textit{profile} is $t=(t_1, t_2, t_3) \in \mathcal{T}={\prod}_{i=1}^3 \mathcal{T}_i$.
  \item $\mathcal{A}_1=\{a_{11}=57,a_{12}=68\}$ denotes bidder $1$'s all possible bids. Similarly, the action sets of the other two bidders are $\mathcal{A}_2=\{a_{21}=70,a_{22}=90\}$ and $\mathcal{A}_3=\{a_{31}=30,a_{32}=80\}$. The real bid \textit{profile} is $a=(a_1, a_2, a_3) \in \mathcal{A}={\prod}_{i=1}^3 \mathcal{A}_i$.
  \item $p(t_1=100,t_2=108,t_3=78)=0.125$ indicates that the probability of the evaluation \textit{profile} $t=(100, 108, 78)$ is $0.125$. The probability distribution of all bidders’ evaluations $p$ is a common knowledge. Thus, the belief
  \begin{align*}
  &p\left(t_2=108, t_3=78 \mid t_1=100\right)\\=&\frac{p\left( t_1=100, t_2=108, t_3=78 \right)}{p\left(t_1=100\right)}
  \end{align*}
  represents bidder $1$'s probability estimation of other bidders' evaluations $(t_2=108, t_3=78)$ when bidder $1$'s evaluation is $t_1=100$. Similarly, other beliefs can be obtained.
  \item $c_i=t_i-a_i$ indicates the payoff function of bidder $i$, that is, the difference between his evaluation and bid is his payoff.
\end{itemize}
\end{example}

Each bidder doesn’t know other bidders’ evaluations of the item. A bidder’s payoff is the difference between his bid and evaluation, so it is closely tied to the evaluation. When modeling, we treat each player’s evaluation as his type which determines his payoff function. If players knew each other’s types, they would know each other’s payoff functions, and the Bayesian game would simplify into a normal-form game. This provides an understanding of the type in \Cref{def3}.

Bidders usually possess common knowledge. So, while they cannot know each other’s exact evaluations, they possess the probability distribution of all bidders’ evaluations. Based on this common knowledge and his own evaluation of the item, each bidder can infer the probability distribution of other bidders' evaluations, which is regarded as his belief. Therefore, a player's belief in \Cref{def3} represents his probability estimation of other players’ types.

In existing definitions of Bayesian games, the action set of each player may be not fixed, but determined by his type, like the payoff function. We combine each player's action sets under different types together to get a fixed action set, so that \Cref{def3} is the same as existing definitions and the notation for the action set is simpler.

\begin{remark}
In this paper, we use \textit{profile} to represent all players' information. For example, each player's action is $a_i \in \mathcal{A}_i$, then the action \textit{profile} is $a=\prod_{i=1}^m a_i \in \mathcal{A}=\prod_{i=1}^m \mathcal{A}_i$.
\end{remark}

\section{Multi-group Bayesian games\label{MBG}}

In this section, we will present the formal definition of multi-group Bayesian games (MBGs) along with its two types of group behaviors: cooperative and non-cooperative behaviors in groups. To characterize equilibria and potentiality under these group behaviors, we subsequently give the following concepts: MBNE and MBPG for cooperative group behavior, strongly MBNE and strongly MBPG for non-cooperative group behavior.

\subsection{The definition of MBGs}

In order to describe the group behavior among players in a Bayesian game, we will give a detailed definition of multi-group Bayesian games in the following. In MBGs, players in the same group play a game with complete information, and players in different groups play a game with incomplete information.

Split players of a Bayesian game into $r$ groups as $G=\mathop{\cup}\limits_{l=1}^{r} G_l$, $G_i \bigcap G_j=\emptyset$, $\forall i\neq j$, $i,j\in \mathcal{D}_{r}$, where $G_l$ denotes the set of players in the $l$-th group. Define $\mid G_l \mid=m_l$. Without loss of generality, denote $G_l=\{\overline{m}_{l-1}+1,\cdots,\overline{m}_{l}\}$, $\overline{m}_{l}={\sum}_{i=1}^{l} m_i$, $l=1,\cdots,r$, $\overline{m}_0=1$.

Define the type set of players in the $l$-th group as $\mathcal{T}_{G_l}={\prod}_{j \in G_l} \mathcal{T}_j$. Then we have $\mathcal{T}={\prod}_{l=1}^r \mathcal{T}_{G_l}$. Define the action set of players in the $l$-th group as $\mathcal{A}_{G_l}={\prod}_{j \in G_l} \mathcal{A}_j $. Then we have $\mathcal{A}={\prod}_{l=1}^r \mathcal{A}_{G_l}$. Based on the above analysis, we give the following definition.

\begin{definition}\label{def4} A multi-group Bayesian game (MBG) is denoted by a tuple $\mathcal{G}=( G,\mathcal{T},\mathcal{A},p,C)$, where
\begin{itemize}[leftmargin=*, topsep=0pt, itemsep=0pt, parsep=0pt]
  \item $G = \cup_{l=1}^{r} G_l$ denotes a player set.
  \item $\mathcal{T}={\prod}_{l=1}^r \mathcal{T}_{G_l}$ indicates a type \textit{profile} set, where $\mathcal{T}_{G_l}={\prod}_{j \in G_l} \mathcal{T}_j$ denotes the type set of players in the $l$-th group.
  \item $\mathcal{A}={\prod}_{l=1}^r \mathcal{A}_{G_l}$ stands for an action \textit{profile} set, where $\mathcal{A}_{G_l}={\prod}_{j \in G_l} \mathcal{A}_j $ denotes the action set of players in the $l$-th group.
  \item \begin{equation}\label{p_T}
  p\left(T_{-l} \mid T_{l}\right)=\frac{p\left(T_{-l},T_{l}\right)}{p\left(T_{l}\right)}=\frac{p\left(T_{-l}, T_{l}\right)}{ \mathop{\sum}\limits_{T_{-l}^{\prime} \in \mathcal{T}_{-G_l}}  p\left(T_{-l}^{\prime}, T_{l}\right)}
  \end{equation}
  represents a belief, where $T_{-l} \in \mathcal{T}_{-G_l}={\prod}_{i \in G_l,i \neq l} \mathcal{T}_{G_i}$, $T_{l} \in \mathcal{T}_{G_l}$. $p: \mathcal{T} \rightarrow[0,1]$ denotes a probability distribution over $\mathcal{T}$, which is a common knowledge. And $T_{l}$ is a private knowledge of players in the $l$-th group.
  \item $c=\{c_1,\cdots,c_m\}$ indicates a payoff function set, where $c_i:\mathcal{T} \times \mathcal{A} \rightarrow \mathbb{R}$ denotes player $i$'s payoff function.
\end{itemize}
\end{definition}

In a Bayesian game, each player adopts different payoff functions under different types, and his type is known only to himself. In the above definition, we treat each group as a whole to indicate that each player's type can only be accurately known by players within his group. So each player only needs to estimate the probability of the type of players in other groups, that is, $p\left(T_{-l} \mid T_{l}\right)$. This aligns closely with real-world scenarios. 

Taking the auction problem as an example, let's illustrate the group behavior involved. Due to the social attributes of the bidders, each bidder knows evaluations (i.e., types) of the bidders within his social circle for the item being auctioned. However, they cannot know the evaluations of bidders outside his social circle. Below we give an example which is based on \Cref{e1} to intuitively show the difference between multi-group Bayesian games and Bayesian games.

\begin{example}\label{e2}
An auction whose auction rule is the first price
sealed bid auction \cite{Elyakime1994} can be modeled as a multi-group Bayesian game and represented by a tuple $\mathcal{G}=( G,\mathcal{T},\mathcal{A},p,c)$, where
\begin{itemize}[leftmargin=*, topsep=0pt, itemsep=0pt, parsep=0pt]
  \item $G_1=\{1\}$ represents group $1$, $G_2=\{2,3\}$ represents group $2$.
  \item $\mathcal{T}_{G_1}=\mathcal{T}_1=\{T_{11}=t_{11},T_{12}=t_{12}\}$ denotes group $1$'s all possible evaluations of the item. Similarly, the type set of group $2$ is $\mathcal{T}_{G_2}={\prod}_{j \in G_2} \mathcal{T}_j=\{T_{21}=(t_{21},t_{31}),T_{22}=(t_{21},t_{32}),T_{23}=(t_{22},t_{31}),T_{24}=(t_{22},t_{32})\}$. The real evaluation \textit{profile} is $T=t=(t_1, t_2, t_3) \in \mathcal{T}={\prod}_{l=1}^2 \mathcal{T}_{G_l}$.
  \item $\mathcal{A}_{G_1}=\mathcal{A}_1=\{A_{11}=a_{11},A_{12}=a_{12}\}$ denotes group $1$'s all possible bids. Similarly, the action set of group $2$ is $\mathcal{A}_{G_2}={\prod}_{j \in G_2} \mathcal{A}_j=\{A_{21}=(a_{21},a_{31}),A_{22}=(a_{21},a_{32}),A_{23}=(a_{22},a_{31}),A_{24}=(a_{22},a_{32})\}$. The real bid \textit{profile} is $A=a=(a_1, a_2, a_3) \in \mathcal{A}={\prod}_{l=1}^2 \mathcal{A}_{G_l}$.
  \item The probability distribution of all bidders’ evaluations $p$ is a common knowledge. $p(t_{11},t_{21},t_{31})=0.125$, $p(t_{11},t_{21},t_{32})\\=0.05$, $p(t_{11},t_{22},t_{31})=0.03$, $p(t_{11},t_{22},t_{32})=0.125$, $p(\\t_{12},t_{21},t_{31})=0.2$, $p(t_{12},t_{21},t_{32})=0.2$, $p(t_{12},t_{22},t_{31})=0.25$, $p(t_{12},t_{22},t_{32})=0.2$. Thus, the belief
  \begin{align*}
  &p\left(T_{11}=100 \mid T_{21}=(t_{21}=108,t_{31}=78)\right)\\=&\frac{p\left(T_{11}=100, T_{21}=(t_{21}=108,t_{31}=78) \right)}{p(T_{21}=(t_{21}=108,t_{31}=78))}
  \end{align*}
  represents group $2$'s probability estimation of group $1$'s evaluations $T_{11}=100$ when group $2$'s evaluation is $T_{21}=(t_{21}=108,t_{31}=78)$. Similarly, other beliefs can be obtained.
  \item $c_i=t_i-a_i$ indicates the payoff function of player $i$.
\end{itemize}
\end{example}

In a MBG, each group selects action based on its type. Therefore, each group's strategy is a mapping from the type \textit{profile} set to the action \textit{profile} set. Here we give the exact definition of each group's strategy in a MBG.

\begin{definition}\label{def5}
In a MBG $\mathcal{G}=( G,\mathcal{T},\mathcal{A},p,C)$, a strategy of the $l$-th group is denoted by $s_{l}: \mathcal{T}_{G_l} \rightarrow \mathcal{A}_{G_l}$. The strategy \textit{profile} set is denoted by $\mathcal{S}=\mathop{\cup}\limits_{l=1}^{r} \mathcal{S}_{G_l}$, where $\mathcal{S}_{G_l}=(\mathcal{A}_{G_l})^{\mathcal{T}_{G_l}}$ denotes the strategy set of the $l$-th group.
\end{definition}

Taking the auction problem as an example, this definition states that each bidder's decision is to determine a bid function $s_{l}$ based on his evaluation $T_l \in \mathcal{T}_{G_l}$, rather than selecting a constant bid unrelated to his evaluation. It should be noted that in the above definition, in order to describe that there is complete information among players within a group, the information of all bidders in the $l$-th group is set as a whole, but in fact, their information just be put together, and each bidder still can choose his strategy independently.

\subsection{Definitions of (strongly) MBNE and (strongly) MBPG}

In actual scenarios, the game with complete information among players within a group may be a cooperative game or a noncooperative game. Still taking the auction problem as an example, there may be two types of relationships between bidders within a social circle. The first one: all bidders within a social circle cooperate with each other. They choose strategies together and share their payoffs equally to reduce risks. The second one: although all bidders  within a social circle are familiar with each other, they do not cooperate with each other. Each bidder still independently chooses his strategy to maximize his own payoff. 

We use \Cref{figbid} to illustrate the difference between Bayesian games and these two situations in multi-group Bayesian games based on \Cref{e1,e2}. And MBGs are applied to a routing problem to further help readers understand the difference between them in \Cref{e3}. 

\begin{figure}[!htb]
\centering
\includegraphics[width=3in]{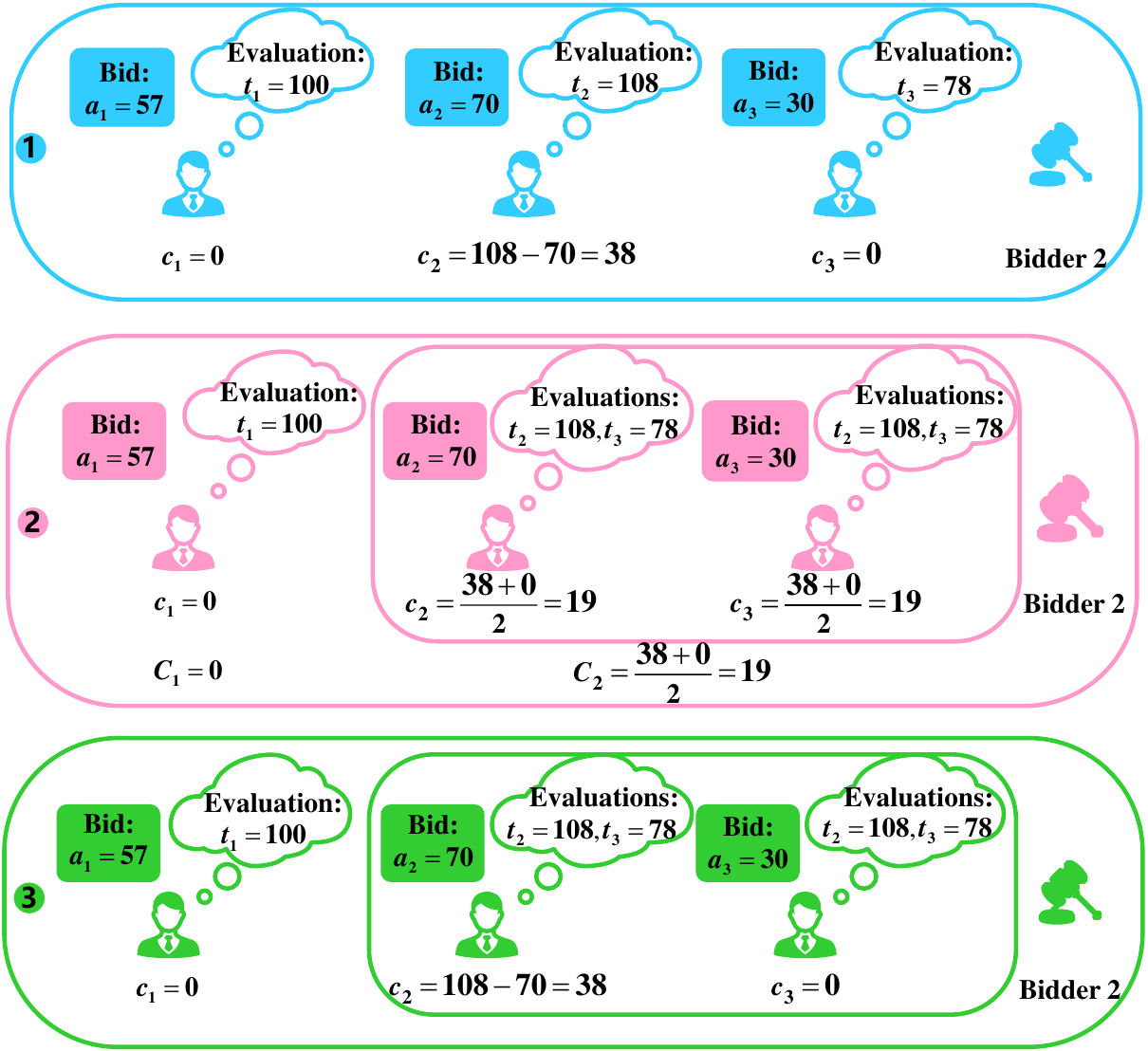}
\caption{Compare the differences between \textcircled{1} a Bayesian game, \textcircled{2} a multi-group Bayesian game where players in a group play a cooperative game, and \textcircled{3} a multi-group Bayesian game where players in a group play a noncooperative game by using an auction problem.}
\label{figbid}
\end{figure}

For the situation where players within a group play a cooperative game with complete information, we define group $l$'s payoff function as $C_{l}\left(T ; A\right)=\frac{\sum_{i \in G_l} c_i\left(T; A\right)}{m_l}$ to represent that players within a group will share their payoffs equally. 
And the goal of group $l$ is to maximize $C_{l}\left(T ; A\right)$. However, since the types of other groups are not precisely known, the goal shifts to maximizing the expected average payoff across all possible types of other groups. Consequently, we introduce the following definitions for MBNE and MBPG.

For the situation where players within a group play a noncooperative game with complete information, we give the following definitions for strongly MBNE and strongly MBPG.

We put the definitions under these two situations together below so that readers can intuitively compare their differences.

\begin{definition}\label{MBNE}
In a MBG $\mathcal{G}=( G,\mathcal{T},\mathcal{A},p,C)$, if a strategy \textit{profile} $ s^{\ast}=( s_{1}^{\ast},\cdots, s_{r}^{\ast}) \in \mathcal{S}$ satisfies that for every group $l \in \mathcal{D}_r$ and any $T_{l} \in \mathcal{T}_{G_l}$,
\begin{enumerate}[leftmargin=*, topsep=0pt, itemsep=0pt, parsep=0pt]
  \item[(i)] $s_{l}^{\ast}(T_{l})$ solves the following maximization problem
$$
\max_{A_{l} \in \mathcal{A}_{G_l}} \sum_{T_{-l} \in \mathcal{T}_{-G_l}} p\left(T_{-l} \mid T_{l}\right)
C_{l}\left(T_{-l}, T_{l} ; s_{-l}^*\left(T_{-l}\right), A_{l}\right),
$$
then $ s^{\ast}$ is called a multi-group Bayesian Nash equilibrium (MBNE);
    \item[(ii)] $s_{l}^{\ast}(T_{l})$ solves the following maximization problem
{\small 
\begin{align}
\max_{A_{l} \in \mathcal{A}_{G_l}} \sum_{T_{-l} \in \mathcal{T}_{-G_l}} &p\left(T_{-l} \mid T_{l}\right)\nonumber \\
&c_{i}\left(T_{-l}, T_{l} ; s_{-l}^*\left(T_{-l}\right), A_{l}\right),\forall i \in G_l
\end{align}
}
then $ s^{\ast}$ is called a strongly multi-group Bayesian Nash equilibrium (strongly MBNE).
\end{enumerate}
\end{definition}

\begin{definition}\label{MBPG}
In a MBG $\mathcal{G}=( G,\mathcal{T},\mathcal{A},p,C)$, if for any $T \in \mathcal{T}$,$A_{l},A_{l}^{\prime} \in \mathcal{A}_{G_l}$, $A_{-l} \in \mathcal{A}_{-G_l}$, $l \in \mathcal{D}_r$, there exists a function $F: \mathcal{T} \times \mathcal{A} \rightarrow \mathbb{R}$ satisfies
  \begin{align}\label{CF}
&C_{l}\left(T ; A_{-l},A_{l}\right)-C_{l}\left(T ;  A_{-l},A_{l}^{\prime}\right)\nonumber \\
=&F\left(T ;  A_{-l},A_{l}\right)-F\left(T ;  A_{-l},A_{l}^{\prime}\right),
\end{align}
then $\mathcal{G}$ is a multi-group Bayesian potential game (MBPG) and $F$ is called a potential function of $\mathcal{G}$. Besides, if (\ref{CF}) is strengthened to 
\begin{align}
&c_{i}\left(T ; A_{-l},A_{l}\right)-c_{i}\left(T ;  A_{-l},A_{l}^{\prime}\right)\nonumber \\
=&F\left(T ;  A_{-l},A_{l}\right)-F\left(T ;  A_{-l},A_{l}^{\prime}\right),\forall i \in G_l,
\end{align}
then $\mathcal{G}$ is a strongly MBPG and $F$ is called a strongly potential function of $\mathcal{G}$.
\end{definition}

Based on the above definitions and knowledge of Bayesian games, we can know that if a MBG is (strongly) potential, then the MBG has (strongly) MBNE. This can be easily extended from the works of \cite{Rosenthal,li2023modeling}. And all (strongly) MBNE are the strategy \textit{profiles} that maximize the (strongly) potential function $F$. 

In the following, MBGs are applied to a routing problem to further help readers understand the difference between Bayesian games and the two situations in multi-group Bayesian games.

\begin{example}\label{e3}
Consider a routing problem. In the traffic network shown in \Cref{lizi1}, players 1,2,3 plan to travel from origin $O$ to destination $D_1$ or $D_2$. Each player gets paid to reach his destination, but, of course, there are costs he needs to pay on his path. And each player's payoff is the difference between these. The type of each player refers to his destination. Therefore, each player’s type determines his payoff function. When a player's destination is $D_1$, his action set is $B_1=\{\beta_1,\beta_2\}$. When a player's destination is $D_2$, his action set is $B_2=\{\gamma_1,\gamma_2\}$. And $\beta_1=(q_1,q_2,q_3)$, $\beta_2=(q_6,q_5,q_3)$, $\gamma_1=(q_6,q_7,q_8)$, $\gamma_2=(q_1,q_4,q_8)$. It is easy to see that the traffic network is symmetrical, so there is a one-to-one correspondence between the actions in $B_1$ and $B_2$. Thus, $\mathcal{A}_i=B$, $i=1,2,3$, where $B=B_1=B_2$. If there is no group behavior among the three players, then the routing problem can be modeled as a Bayesian game. 

However, if there exists group behavior between the three players, then the routing problem can be modeled as a multi-group Bayesian game. Suppose that players 1 and 2 are in a group, then they know each other's destination. So the game between them, that is, the game within the group, is with complete information. Player 3 is not familiar with them and don't know their destinations, so the game between player 3 and players 1, 2 is with incomplete information. But the prior distribution of the destinations chosen by all players denoted by $p$ is public. 

There are two types of situations. The first one: all players within a group play a cooperative game with complete information, that is, they choose paths together and share their payoffs equally to reduce risks. The second one: all players within a group play a noncooperative game with complete information, that is, although all players within a group know each other's destination, each player still independently chooses his path to maximize his own payoff. 

\begin{figure}[htb]
\centering
\includegraphics[width=3in]{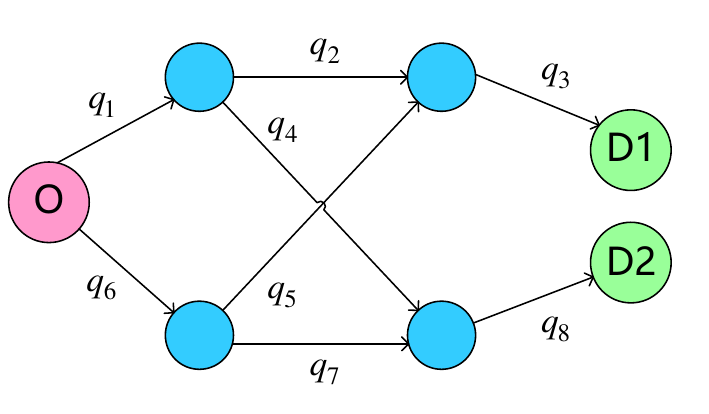}
\caption{The traffic network.}
\label{lizi1}
\end{figure}
\end{example}

\begin{remark}
For games with complete information, \cite{Li2019} put forward the concept of ``strongly'', indicating that the requirement is stronger, so we also use this concept in this paper to align with existing concepts in games with complete information.
\end{remark}

\subsection{Definition and properties of the multi-group ex-ante agent
transformation\label{MEAT}}

For finding (strongly) MBNE of a MBG, we will transform a MBG to its MEAG which is more tractable. Subsequently, we establish a bijective correspondence between (strongly) MBNE of the original MBG and (strongly) NE of the transformed MEAG, ensuring that a MBG's (strongly) MBNE can be obtained directly by solving its MEAG's (strongly) NE. Furthermore, we demonstrate that this transformation preserves (strongly) potentiality, so even if a MBG is not potential, as long as its MEAG is potential, its (strongly) MBNE can be obtained.

Firstly, we give some concepts about a MBG's multi-group ex-ante agent game (MEAG) and provide the relationship between a MBG and its MEAG in this section.

\begin{definition}\label{def8}
A multi-group ex-ante agent game (MEAG) of a MBG $\mathcal{G}=( G,\mathcal{T},\mathcal{A},p,C)$ is represented by a tuple $\widehat{\mathcal{G}}=\{\widehat{G}, \widehat{\mathcal{A}}, \widehat{C}\}$, where
\begin{itemize}[leftmargin=*, topsep=0pt, itemsep=0pt, parsep=0pt, partopsep=0pt]
  \item Agent set: $\widehat{G}=\{(l,T_{lk})\}_{l\in \mathcal{D}_{r},T_{lk}\in \mathcal{T}_{G_l}}$, where agent $(l,T_{lk})$ represents the type of group $l$ is $T_{lk}$, that is, the $k$-th type.
  \item Agent action \textit{profile} set: $\widehat{\mathcal{A}}={\prod}_{(l,T_{lk})\in \widehat{G}} \widehat{\mathcal{A}}_{(l,T_{lk})}$, where $\widehat{\mathcal{A}}_{(l,T_{lk})}= \mathcal{A}_{G_l}$ stands for the action set of agent $(l,T_{lk}) \in \widehat{G}$.
  \item Payoff function: $\widehat{c}=\{\widehat{c}_{(l,T_{lk})}\}_{(l,T_{lk}) \in \widehat{G}}$, where $\widehat{c}_{(l,T_{lk})}=\{\widehat{c}_{(l,T_{lk})}^i\}_{i \in G_l}$,
      \begin{align}\label{C_jian}
      \hat{c}_{(l,T_{lk})}^i(\widehat{A})=&\sum_{T_{-l} \in \mathcal{T}_{-{G_l}}} p(T_{-l}, T_{lk}) \nonumber \\
      & c_i(T_{-l}, T_{l k} ;[\Gamma^{-1}(\widehat{A})](T_{-l}, T_{lk}))
      \end{align}
      represents player $i$'s payoff function in agent $(l,T_{lk})$.
\end{itemize}
\end{definition}

In this definition, each type of each group $(l,T_{lk})$ is regarded as an agent, that is, a new player. And beliefs are incorporated into payoff functions. So the MEAG of a MBG no longer contains incomplete information and becomes a normal-form game.

To describe the one-to-one relationship between strategies in a MBG and agent action \textit{profiles} in the MBG's MEAG, define a bijective mapping $\Gamma: \mathcal{S} \rightarrow \widehat{\mathcal{A}}$ as
\begin{align}
&\Gamma(s)=\left(s_l\left(T_{l k}\right)\right)_{l \in \mathcal{D}_r, T_{l k} \in \mathcal{T}_{G_l}} \nonumber \\
=&(s_1(T _{11}), \ldots, s_1(T_{1 \mid \mathcal{T}_{G_1} \mid }), \ldots, s_r(T_{r 1}), \ldots, s_r(T_{r \mid \mathcal{T}_{G_r} \mid })),
\end{align}
for any $s=\left(s_1, \ldots, s_r\right) \in \mathcal{S}$. Then we can get that for any $ \widehat{A}=(\widehat{A}_{(l, T_{l k})})_{(l, T_{lk}) \in \widehat{G}} \in \widehat{\mathcal{A}} $, $\Gamma$'s inverse mapping $\Gamma^{-1}$ satisfies $ \Gamma^{-1}(\widehat{A}) \in \mathcal{S}$. So $[\Gamma^{-1}(\widehat{A})]_l (T_l)$$=\widehat{A}_{(l, T_l)} \in  \mathcal{A}_{G_l}$.

For any $\widehat{A}\in \widehat{\mathcal{A}}$, define agent $\left(l, T_{lk}\right) $'s payoff function as $ \widehat{C}_{(l,T_{lk})}(\widehat{A})=\frac{\sum_{i \in G_l}  \widehat{c}_{(l,T_{lk})}^i(\widehat{A})}{m_l}$ to represent that players within a group will share their payoffs equally. 

\begin{definition}\label{def9}
For a MEAG $\widehat{\mathcal{G}}=\{\widehat{G}, \widehat{\mathcal{A}}, \widehat{C}\}$, an agent action \textit{profile} $\widehat{A}^*=( \widehat{A} _{\left(l, T_{lk}\right)}^*, \widehat{A}_{-\left(l, T_{lk}\right)}^*)$ satisfies that:

\begin{enumerate}[leftmargin=*, topsep=0pt, itemsep=0pt, parsep=0pt]
  \item[(i)] $\widehat{A}^*$ is a Nash equilibrium if for every $\left(l, T_{lk}\right) \in \widehat{G}$, $\widehat{A} _{\left(l, T_{lk}\right)}^*$ solves the following maximization problem
$$
\max_{\widehat{A} _{\left(l, T_{lk}\right)} \in \widehat{\mathcal{A}}_{(l,T_{lk})}}\widehat{C}_{(l,T_{lk})}(\widehat{A} _{\left(l, T_{lk}\right)}, \widehat{A}_{-\left(l, T_{lk}\right)}^*).
$$
  \item[(ii)] $\widehat{A}^*$ is a strongly Nash equilibrium if for every $\left(l, T_{lk}\right) \in \widehat{G}$, $\widehat{A} _{\left(l, T_{lk}\right)}^*$ solves the following maximization problem
$$
\max_{\widehat{A} _{\left(l, T_{lk}\right)} \in \widehat{\mathcal{A}}_{(l,T_{lk})}}\widehat{c}_{(l,T_{lk})}^i(\widehat{A} _{\left(l, T_{lk}\right)}, \widehat{A}_{-\left(l, T_{lk}\right)}^*) ,\forall i \in G_l.
$$
\end{enumerate}
\end{definition}

Secondly, as in \Cref{fig1}, we give the following two theorems to show the proposed transformation's good properties.

\begin{figure}[!htb]
\centering
\includegraphics[width=3in]{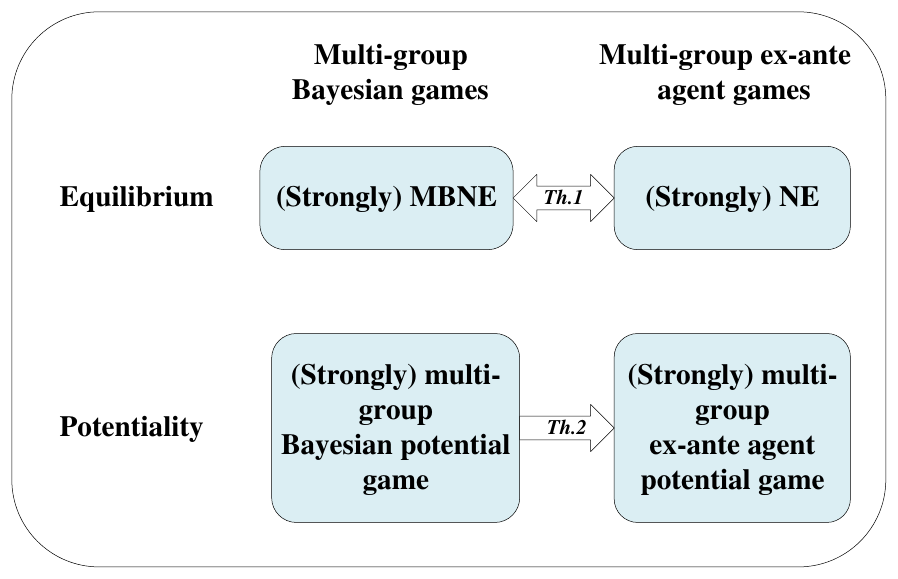}
\caption{The properties of multi-group ex-ante agent transformation.}
\label{fig1}
\end{figure}

\begin{theorem}\label{th:thm1}
For a MBG $\mathcal{G}=( G,\mathcal{T},\mathcal{A},p,C)$ and its MEAG $\widehat{\mathcal{G}}=\{\widehat{G}, \widehat{\mathcal{A}}, \widehat{C}\}$, $s^*$ is a (strongly) multi-group Bayesian Nash equilibrium of $\mathcal{G}$ if and only if $\widehat{A}^*=\Gamma(s^*)$ is a (strongly) Nash equilibrium of $\widehat{\mathcal{G}}$.
\end{theorem}

\begin{proof}
We assume that $\widehat{A}^*$ is a Nash equilibrium of $\widehat{\mathcal{G}}$. By Definitions \ref{def8}, \ref{def9}, we can get that for every $\left(l, T_{lk}\right) \in \widehat{G}$, $\widehat{A} _{\left(l, T_{lk}\right)}^*=[\Gamma^{-1}(\widehat{A}^*)]_l (T_{lk})$ solves
\begin{align}\label{eq7}
\max _{\widehat{A} _{\left(l, T_{lk}\right)} \in \widehat{\mathcal{A}}_{(l,T_{lk})}} &\sum_{T_{-l} \in \mathcal{T}_{-{G_l}}}p(T_{-l}, T_{lk}) \nonumber \\
& c_i(T_{-l}, T_{l k} ;[\Gamma^{-1}(\widehat{A}^*)]_{-l}(T_{-l}),\widehat{A} _{\left(l, T_{lk}\right)}),\forall i \in G_l.
\end{align}

Because $\widehat{A}^*=\Gamma(s^*)$, we can obtain that
$$s_l^*(T_{l})=[\Gamma^{-1}(\widehat{A}^*)]_l (T_{l}),s_{-l}^*(T_{-l})=[\Gamma^{-1}(\widehat{A}^*)]_{-l} (T_{-l}).$$
Then we have
$$s_l^*(T_{lk})=[\Gamma^{-1}(\widehat{A}^*)]_l (T_{lk})=\widehat{A} _{\left(l, T_{lk}\right)}^*.$$

So (\ref{eq7}) can be rewritten as for every group $l \in \mathcal{D}_r$ and any $T_{lk} \in \mathcal{T}_{G_l}$, $s_{l}^{\ast}(T_{lk})$ solves
\begin{align}\label{eq8}
\max_{\widehat{A} _{\left(l, T_{lk}\right)} \in \widehat{\mathcal{A}}_{(l,T_{lk})}} &\sum_{T_{-l} \in \mathcal{T}_{-G_l}} p\left(T_{-l} , T_{lk}\right)\nonumber \\
&c_{i}(T_{-l}, T_{lk} ; s_{-l}^*\left(T_{-l}\right), \widehat{A} _{\left(l, T_{lk}\right)}),\forall i \in G_l.
\end{align}
Since $p\left(T_{-l} , T_{lk}\right)=p\left(T_{-l} \mid T_{lk}\right) p\left(T_{lk}\right)$ and $\widehat{\mathcal{A}}_{(l,T_{lk})}= \mathcal{A}_{G_l}$, (\ref{eq7}) can also be represented by
\begin{align}
\max_{A_{l} \in \mathcal{A}_{G_l}} &\sum_{T_{-l} \in \mathcal{T}_{-G_l}} p\left(T_{-l} \mid T_{lk}\right) \nonumber \\
&c_{i}(T_{-l}, T_{lk} ; s_{-l}^*\left(T_{-l}\right), A_{l}),\forall i \in G_l.
\end{align}
so $s^*$ is a multi-group Bayesian Nash equilibrium of $\mathcal{G}$. And vice versa. So this completes the proof.
\end{proof}

\begin{theorem}\label{th:thm2}
If a MBG $\mathcal{G}=( G,\mathcal{T},\mathcal{A},p,C)$ is (strongly) potential, then its MEAG $\widehat{\mathcal{G}}=\{\widehat{G}, \widehat{\mathcal{A}}, \widehat{C}\}$ is (strongly) potential.
\end{theorem}

\begin{proof}
We denote a potential function of the multi-group Bayesian potential game $\mathcal{G}$ as $F$. We construct the function $\widehat{F}:\widehat{\mathcal{A}} \rightarrow \mathbb{R}$ as
$$
\widehat{F}=\sum_{T \in \mathcal{T}} p\left(T\right)
F(T ; \Gamma^{-1}(\widehat{A})(T)),\forall \widehat{A} \in \widehat{\mathcal{A}}.
$$

For any agent $(l,T_{lk}) \in \widehat{G}$ and any $\widehat{A}=( \widehat{A} _{\left(l, T_{lk}\right)}, \widehat{A}_{-\left(l, T_{lk}\right)}), \widehat{A}^{\prime}=( \widehat{A}^{\prime}_{\left(l, T_{lk}\right)}, \widehat{A}_{-\left(l, T_{lk}\right)})\in \widehat{\mathcal{A}}$ in $\mathcal{G}$ 's MEAG $\widehat{\mathcal{G}}$, we can obtain that
\begin{align}\label{eq10}
& \widehat{F}(\widehat{A})-\widehat{F}(\widehat{A}^{\prime}) \nonumber \\
& =\sum_{T_{-l} \in \mathcal{T}_{-G_l}} p(T_{-l}, T_{lk})[F(T_{-l}, T_{lk} ;[\Gamma^{-1}(\widehat{A})]_{-l}(T_{-l}), \widehat{A} _{(l, T_{lk})})\nonumber \\
& -F(T_{-l}, T_{lk} ;[\Gamma^{-1}(\widehat{A})]_{-l}(T_{-l}), \widehat{A}^{\prime}_{(l, T_{lk})})] \nonumber \\
& +\sum_{T_{-l} \in \mathcal{T}_{-G_l}}  \sum_{T_{l} \in \mathcal{T}_{G_l}}^{T_{l} \neq T_{lk}} p(T_{-l}, T_{l})[F(T_{-l}, T_{l} ;[\Gamma^{-1}(\widehat{A})](T_{-l},T_{l})) \nonumber \\
& -F(T_{-l}, T_{l} ;[\Gamma^{-1}(\widehat{A}^{\prime})](T_{-l},T_{l}))]
\end{align}

For any $T_{-l} \in \mathcal{T}_{-G_l}$ and $T_{l} \in \mathcal{T}_{G_l}$ with $T_{l} \neq T_{lk}$, we have $[\Gamma^{-1}(\widehat{A})](T_{-l},T_{l}))=[\Gamma^{-1}(\widehat{A}^{\prime})](T_{-l},T_{l}))$, so the second term of (\ref{eq10}) equals $0$. Because $F$ is a potential function of the multi-group Bayesian potential game $\mathcal{G}$, we can obtain that
\begin{align}\label{eq11}
&F(T_{-l}, T_{lk} ;[\Gamma^{-1}(\widehat{A})]_{-l}(T_{-l}), \widehat{A} _{(l, T_{lk})})\nonumber \\
& -F(T_{-l}, T_{lk} ;[\Gamma^{-1}(\widehat{A})]_{-l}(T_{-l}), \widehat{A}^{\prime}_{(l, T_{lk})}) \nonumber \\
=&c_i(T_{-l}, T_{lk} ;[\Gamma^{-1}(\widehat{A})]_{-l}(T_{-l}), \widehat{A} _{(l, T_{lk})})\nonumber \\
& -c_i(T_{-l}, T_{lk} ;[\Gamma^{-1}(\widehat{A})]_{-l}(T_{-l}), \widehat{A}^{\prime}_{(l, T_{lk})}),\forall i \in G_l.
\end{align}
Substitute (\ref{eq11}) into (\ref{eq10}), then we have
\begin{align}\label{eq12}
& \widehat{F}(\widehat{A})-\widehat{F}(\widehat{A}^{\prime}) \nonumber \\
=& \sum_{T_{-l} \in \mathcal{T}_{-G_l}} p(T_{-l}, T_{lk})[c_i(T_{-l}, T_{lk} ;[\Gamma^{-1}(\widehat{A})]_{-l}(T_{-l}), \widehat{A} _{(l, T_{lk})})\nonumber \\
& -c_i(T_{-l}, T_{lk} ;[\Gamma^{-1}(\widehat{A})]_{-l}(T_{-l}), \widehat{A}^{\prime}_{(l, T_{lk})})] \nonumber \\
 =&\hat{c}_{(l,T_{lk})}^i(\widehat{A})-\hat{c}_{(l,T_{lk})}^i(\widehat{A}^{\prime}),\forall i \in G_l.
\end{align}

So $\widehat{F}$ is a potential function of $\widehat{\mathcal{G}}$. Then we can get that $\widehat{\mathcal{G}}$ is a potential game. This completes the proof.
\end{proof}

Finally, according to the above theorems, we can get the following essential conclusion. For every (strongly) multi-group Bayesian potential game, we can obtain its (strongly) multi-group Bayesian Nash equilibria by finding (strongly) Nash equilibria of its MEAG. What's more, for MBGs which are not (strongly) potential, as long as their MEAGs are (strongly) potential, this method will still work. Therefore, this transformation has very good properties.

As can be seen from \Cref{MBNE}, we need to find (strongly) MBNE of MBGs on a mapping space, which is very difficult, so we need to use this transformation with good properties to transform MBGs into normal-form games, whose (strongly) Nash equilibria are no longer mappings, but specific action \textit{profiles}. For a normal-form game, the method of using a potential equation to find its (strongly) Nash equilibria can be referred to \cite{cheng2014,Li2019}.

\section{Verification of potentiality and algorithms of finding (strongly) MBNE\label{Ver}}

In the above, we have given an overall framework for finding a MBG's (strongly) MBNE. We will give this method's specific steps in the following and provide algorithms of finding (strongly) MBNE.

\Cref{first} offers the algebraic forms of two functions that are related to (strongly) Nash equilibria of a MEAG. Based on the obtained algebraic forms, \Cref{second} gives the (strongly) potential equation of a MBG's MEAG to judge whether the MEAG is (strongly) potential. Based on the above, \Cref{third} gives algorithms and examples for finding (strongly) MBNE in MBGs whose MEAGs are (strongly) potential. 

\subsection{Algebraic form of the multi-group ex-ante agent game\label{first}}

According to \Cref{def9}, it is evident that Nash equilibria of a MEAG are closely related to each agent $(l,T_{lk})$'s payoff function $\widehat{C}_{(l,T_{lk})}(\widehat{A})$ and strongly Nash equilibria are closely related to each player's payoff function $\widehat{c}_{(l,T_{lk})}^i(\widehat{A})$, so we obtain the algebraic forms of these two functions in this subsection.

Firstly, we provide an introduction to the algebraic form of a pseudo-logical function. Intuitively, solving the algebraic form of a function is separating its variables from the remaining elements(denoted as the structure matrix) based on the properties of the STP.

Let's combine \Cref{lem2} to explain the algebraic form and structure matrix of a pseudo-logical function in detail. Every pseudo-logical function can be written as the product of two parts, the first being its structure matrix which is a row vector and the second being the STP of all variables which are column vectors. The variables and their STP are logical vectors, that is, column vectors with only one element being 1 and the rest being 0. 

There are two methods to solve the algebraic form of a function:
\begin{enumerate}[leftmargin=*, topsep=0pt, itemsep=0pt, parsep=0pt]
  \item[(i)] For a simple function, its algebraic form can be obtained by lexicographically ordering all possible values of variables and then calculating the corresponding function values into the structure matrix.
  \item[(ii)] For a complex function, its algebraic form is derived from the structure matrices of some simple functions by using the STP.
\end{enumerate}

Secondly, since $\widehat{C}_{(l,T_{lk})}(\widehat{A})$ and $\widehat{c}_{(l,T_{lk})}^i(\widehat{A})$ are computed by player $i$'s payoff function $c_i(t;a)$ in group $l$ and the probability distribution function $p(t)$ over $\mathcal{T}$, we use the first method to solve the algebraic forms of $c_i(t;a)$ and $p(t)$ first.

Given a MBG $\mathcal{G}=( G,\mathcal{T},\mathcal{A},p,C)$, without loss of generality, denote $\mid \mathcal{A}_i \mid=g$, $\mid \mathcal{T}_i \mid=e$, $i \in G$, $\mid \mathcal{A}_{G_l} \mid=g^{m_l}:=g_l$, $\mid \mathcal{T}_{G_l} \mid=e^{m_l}:=e_l$, $l\in \mathcal{D}_r$. We arrange type \textit{profiles} in lexicographical order as
\begin{align}\label{eq12}
& t^1=(t_{11},t_{21},\cdots,t_{m1}), t^2=(t_{11},t_{21},\cdots,t_{m2}), \cdots \nonumber \\
& t^{e^m-1}=(t_{1e},t_{2e},\cdots,t_{m[e-1]}), t^{e^m}=(t_{1e},t_{2e},\cdots,t_{me}). \nonumber
\end{align}
Based on the probability distribution $p: \mathcal{T} \rightarrow[0,1]$, we define a probability vector as
\begin{equation}\label{P}
P=[p(t^1),p(t^2),\cdots,p(t^{e^m})].
\end{equation}
Identify $ \mathcal{T}_i \sim \Delta_e$, $ \mathcal{A}_i \sim \Delta_g$. Then for any $t=(t_1,\cdots,t_m) \in \mathcal{T}$ and any $a=(a_1,\cdots,a_m) \in \mathcal{A}$, we identify $t \sim \ltimes_{i=1}^m t_i \in \Delta_{e^m}$, $a \sim \ltimes_{i=1}^m a_i \in \Delta_{g^m}$. Therefore, by \Cref{lem2}, we can get
\begin{equation}\label{pt}
p(t)=P \ltimes t,
\end{equation}
\begin{equation}\label{cita1}
c_i(t;a)=\mathbb{C}_i \ltimes t \ltimes a,
\end{equation}
where $P$ is defined by \Cref{P}, and $\mathbb{C}_i \in \mathbb{R}_{1 \times (eg)^m}$ contains all possible values of $c_i(t;a)$. Define
\begin{equation}\label{mathbbC}
\mathbb{C}=[\mathbb{C}_1,\mathbb{C}_2,\cdots, \mathbb{C}_m],
\end{equation}
then \Cref{cita1} can be expressed as
\begin{equation}\label{cita2}
c_i(t;a)=\mathbb{C} \ltimes \delta_m^i \ltimes t \ltimes a.
\end{equation}

Finally, using the algebraic forms obtained above, we use the second method to solve the algebraic forms of $\widehat{C}_{(l,T_{lk})}(\widehat{A})$ and $\widehat{c}_{(l,T_{lk})}^i(\widehat{A})$, thereby laying a foundation for finding (strongly) Nash equilibria of a MEAG by using potential functions. 

Since $\widehat{C}_{(l,T_{lk})}(\widehat{A})$ and $\widehat{c}_{(l,T_{lk})}^i(\widehat{A})$ are functions of the variables $\widehat{A}$ and $T$, where $T=t$ and $\widehat{A}$ is a function of $t$ and $a$, we use the following proposition to change the variables of $c_i$ from $a$, $t$ to $\widehat{A}$, $T$. That is to say, we will give the algebraic form of $c_i$ with respect to $\widehat{A}$ and $T$. 

\begin{proposition}\label{pro2}
In the MEAG $\widehat{\mathcal{G}}=\{\widehat{G}, \widehat{\mathcal{A}}, \widehat{C}\}$ of a MBG $\mathcal{G}=( G,\mathcal{T},\mathcal{A},p,C)$, if the action \textit{profile} $\widehat{A} \in \widehat{\mathcal{A}}$ is given, then under an arbitrary fixed type \textit{profile} $T\in  \mathcal{T} $, player $i$'s payoff function in $\mathcal{G}$ is
\begin{equation}\label{ci}
c_i(T;\Gamma^{-1}(\widehat{A})(T))=\mathbb{C} \ltimes {\delta}_m^i \ltimes \Phi(\widehat{A}) \ltimes T,
\end{equation}
where $\Phi(\widehat{A})=I_{\widetilde{e}} \otimes [{\otimes}_{l=1}^r E(\widehat{A},l)] \ltimes O_{\widetilde{e}}$.
\end{proposition}

\begin{proof} 
Identify $\widehat{\mathcal{A}}_{(l,T_{lk})} \sim \Delta_{g_l}$. Then we can get $ \widehat{A}=(\widehat{A}_{(1, T_{1 1})},\cdots,\widehat{A}_{(1, T_{1 e_1})},\cdots,\widehat{A}_{(r, T_{r 1})},\cdots,\widehat{A}_{(r, T_{r e_r})}) \sim \ltimes_{l=1}^r \ltimes_{k=1}^{e_l} \widehat{A}_{(l, T_{l k})} \in \Delta_{\widetilde{h}}$, where $\widetilde{h}={\prod}_{l=1}^r h_l$, $h_l=g_l^{e_l}$.

By virtue of Definition \ref{def1} and Lemma \ref{lem1}, we can easily get that
\begin{align}
\Gamma^{-1}(\widehat{A})_l(T_{lk_l})=\widehat{A}_{(l, T_{l k_l})}
=M_{[g_l^{k_l-1},g_l,g_l^{e_l-k_l}]} M_{[\widetilde{h}_{1l},h_l,\widetilde{h}_{lr}]}\widehat{A},
 \nonumber
\end{align}
where $\widetilde{h}_{1l}={\prod}_{j=1}^{l-1} h_j$, $\widetilde{h}_{lr}={\prod}_{j=l+1}^{r} h_j$. Define $E(\widehat{A},l)=[M_{[1,g_l,g_l^{e_l-1}]} M_{[\widetilde{h}_{1l},h_l,\widetilde{h}_{lr}]}\widehat{A},\cdots,M_{[g_l^{e_l-1},g_l,1]} M_{[\widetilde{h}_{1l},h_l,\widetilde{h}_{lr}]}\widehat{A}]$, then $\Gamma^{-1}(\widehat{A})_l(T_{l})=E(\widehat{A},l) T_{l}$.

Therefore, by using Lemma \ref{lem1}, we can get the vector form of $\Gamma^{-1}(\widehat{A})(T)=(\Gamma^{-1}(\widehat{A})_1(T_{1}),\cdots,\Gamma^{-1}(\widehat{A})_r(T_{r}))$ as
\begin{align}
 &\Gamma^{-1}(\widehat{A})(T)={\ltimes}_{l=1}^{r}\Gamma^{-1}(\widehat{A})_l(T_{l}) \nonumber \\
=&{\ltimes}_{l=1}^{r} E(\widehat{A},l) T_{l}=[{\otimes}_{l=1}^r E(\widehat{A},l)]{\ltimes}_{l=1}^r T_{l}. \nonumber
\end{align}

By virtue of (\ref{cita2}), we can obtain that player $i$'s payoff function is
\begin{align}
&c_i(T;\Gamma^{-1}(\widehat{A})(T))=\mathbb{C} \ltimes {\delta}_m^i  \ltimes T \ltimes \Gamma^{-1}(\widehat{A})(T) \nonumber \\
=&\mathbb{C} \ltimes {\delta}_m^i  \ltimes T \ltimes [{\otimes}_{l=1}^r E(\widehat{A},l)] \ltimes T \nonumber \\
=&\mathbb{C} \ltimes {\delta}_m^i  \ltimes I_{\widetilde{e}} \otimes [{\otimes}_{l=1}^r E(\widehat{A},l)] \ltimes O_{\widetilde{e}} \ltimes T, \nonumber
\end{align}
where $\widetilde{e}={\prod}_{l=1}^r e_l$.
\end{proof} 

Based on \Cref{ci}, we can derive each player's payoff function $\widehat{c}_{(l,T_{lk})}^i(\widehat{A})$ and each agent $(l,T_{lk})$'s payoff function $\widehat{C}_{(l,T_{lk})}(\widehat{A})$, as shown in the following theorem.

\begin{theorem}\label{the3}
In the MEAG $\widehat{\mathcal{G}}=\{\widehat{G}, \widehat{\mathcal{A}}, \widehat{C}\}$ of a MBG $\mathcal{G}=( G,\mathcal{T},\mathcal{A},p,C)$, if the action \textit{profile} $\widehat{A}\in \widehat{\mathcal{A}}$ is given, then player $i$'s payoff function in agent $(l,T_{lk})$ is termed as
\begin{equation}\label{c_jian_i}
\widehat{c}_{(l,T_{lk})}^i(\widehat{A})
= P \ltimes \Xi(l,k) \ltimes \Phi^{T}(\widehat{A}) \ltimes  ({\delta}_m^i)^{T} \ltimes {\mathbb{C}}^{T}
\end{equation}
and agent $\left(l, T_{lk}\right) $'s payoff function is represented by
\begin{equation}\label{c_jian_l}
\widehat{C}_{(l,T_{lk})}(\widehat{A})
= \frac{1}{m_l} P \ltimes \Xi(l,k) \ltimes \Phi^{T}(\widehat{A}) \ltimes \sum_{i \in G_l} ({\delta}_m^i)^{T} \ltimes {\mathbb{C}}^{T},
\end{equation}
where $\Xi(l,k)=I_{\widetilde{e}_{1l}} \otimes ({\delta}_{e_l}^k({\delta}_{e_l}^k)^{T})$.
\end{theorem}

\begin{proof} 
By using Lemma \ref{lem1}, we have
\begin{align}\label{T}
T=&T_{1} \ltimes \ldots \ltimes T_{l-1} \ltimes T_{l} \ltimes T_{l+1} \ltimes \ldots \ltimes T_{r} \nonumber \\
=&W_{[e_l,\widetilde{e}_{1l}]} \ltimes T_{l} \ltimes T_{1} \ltimes \ldots \ltimes T_{l-1} \ltimes  T_{l+1} \ltimes \ldots \ltimes T_{r} \nonumber \\
=&W_{[e_l,\widetilde{e}_{1l}]} \ltimes T_{l} \ltimes T_{-l},
\end{align}
where $\widetilde{e}_{1l}={\prod}_{j=1}^{l-1} e_j$. Then (\ref{pt}) can be rewritten as
$$p(T_{l},T_{-l})=P \ltimes T=P \ltimes W_{[e_l,\widetilde{e}_{1l}]} \ltimes T_{l} \ltimes T_{-l}.
$$
When $T_{l}=T_{lk}={\delta}_{e_l}^k$, we have
$$p(T_{lk},T_{-l})=P \ltimes W_{[e_l,\widetilde{e}_{1l}]} \ltimes {\delta}_{e_l}^k \ltimes T_{-l}.
$$

According to Lemma \ref{lem2}, we obtain the structure vector of function $p(T_{lk},T_{-l})$ as
$$
L_{p(T_{lk},T_{-l})}=P \ltimes W_{[e_l,\widetilde{e}_{1l}]} \ltimes {\delta}_{e_l}^k.
$$

For every player $i$ in group $l$, by virtue of (\ref{ci}) and (\ref{T}), we can get
\begin{equation}\label{ci_new}
c_i(T;\Gamma^{-1}(\widehat{A})(T))=\mathbb{C} \ltimes {\delta}_m^i \ltimes \Phi(\widehat{A}) \ltimes W_{[e_l,\widetilde{e}_{1l}]} \ltimes T_{l} \ltimes T_{-l}.
\end{equation}
When $T_{l}=T_{lk}={\delta}_{e_l}^k$, we have
$$
c_i(T;\Gamma^{-1}(\widehat{A})(T_{lk},T_{-l}))=\mathbb{C} \ltimes {\delta}_m^i \ltimes \Phi(\widehat{A}) \ltimes W_{[e_l,\widetilde{e}_{1l}]} \ltimes {\delta}_{e_l}^k \ltimes T_{-l}.
$$

According to Lemma \ref{lem2}, we obtain the structure vector of function $c_i(T;\Gamma^{-1}(\widehat{A})(T_{lk},T_{-l}))$ as
$$
L_{c_i(T;\Gamma^{-1}(\widehat{A})(T_{lk},T_{-l}))}=\mathbb{C} \ltimes {\delta}_m^i \ltimes \Phi(\widehat{A}) \ltimes W_{[e_l,\widetilde{e}_{1l}]} \ltimes {\delta}_{e_l}^k.
$$

Using Corollary \ref{cor1}, (\ref{C_jian}) can be rewritten as
\begin{align}\label{C_jian_new}
&\hat{c}_{(l,T_{lk})}^i(\widehat{A})\nonumber \\
=&\sum_{T_{-l} \in \mathcal{T}_{-{G_l}}} p(T_{lk},T_{-l}) c_i(T_{lk},T_{-l}; [\Gamma^{-1}(\widehat{A})](T_{lk},T_{-l})) \nonumber \\
       =&L_{p(T_{lk},T_{-l})}L^{T}_{c_i(T;\Gamma^{-1}(\widehat{A})(T_{lk},T_{-l}))} \nonumber \\
=&P \ltimes W_{[e_l,\widetilde{e}_{1l}]} \ltimes {\delta}_{e_l}^k  ({\delta}_{e_l}^k)^{T}  \ltimes W_{[e_l,\widetilde{e}_{1l}]}^{T} \ltimes \Phi(\widehat{A})^{T} \ltimes ({\delta}_m^i)^{T} \ltimes \mathbb{C}^{T}.
\end{align}

Using Lemma \ref{lem1}, we have
\begin{align}\label{sum}
&W_{[e_l,\widetilde{e}_{1l}]} \ltimes {\delta}_{e_l}^k  ({\delta}_{e_l}^k)^{T}  \ltimes W_{[e_l,\widetilde{e}_{1l}]}^{T} \nonumber \\
       =&W_{[e_l,\widetilde{e}_{1l}]} \ltimes {\delta}_{e_l}^k  ({\delta}_{e_l}^k)^{T}  \ltimes W_{[\widetilde{e}_{1l},e_l]} \nonumber \\
=&I_{\widetilde{e}_{1l}} \otimes ({\delta}_{e_l}^k({\delta}_{e_l}^k)^{T}):=\Xi(l,k) .
\end{align}

Then (\ref{C_jian_new}) can be rewritten as
\begin{align}\label{C_jian_neww}
\hat{c}_{(l,T_{lk})}^i(\widehat{A})=P \ltimes \Xi(l,k) \ltimes \Phi(\widehat{A})^{T} \ltimes ({\delta}_m^i)^{T} \ltimes \mathbb{C}^{T}. \nonumber
\end{align}
And agent $\left(l, T_{lk}\right) $'s payoff function can be represented by
\begin{align*}
\hat{C}_{(l,T_{lk})}(\widehat{A})=&\frac{\sum_{i \in G_l}  \hat{c}_{(l,T_{lk})}^i(\widehat{A})}{m_l}\\
=&\frac{1}{m_l} P \ltimes \Xi(l,k) \ltimes \Phi^{T}(\widehat{A}) \ltimes \sum_{i \in G_l} ({\delta}_m^i)^{T} \ltimes {\mathbb{C}}^{T}.
\end{align*}

The prove is completed.
\end{proof} 

It is not difficult to find that in (\ref{c_jian_i}) and,(\ref{c_jian_l}, variable $\widehat{A}=\mathop{\ltimes}\limits_{l=1}^r  \mathop{\ltimes}\limits_{k=1}^{e_l} \widehat{A}_{{(l,T_{lk})}}$ is not separated, so in order to obtain the algebraic form of $\widehat{c}_{(l,T_{lk})}^i(\widehat{A})$ and $\widehat{C}_{(l,T_{lk})}(\widehat{A})$, we give the following proposition.

\begin{proposition}\label{pro3}
In a MEAG $\widehat{\mathcal{G}}=\{\widehat{G}, \widehat{\mathcal{A}}, \widehat{C}\}$, the algebraic form of player $i$'s payoff function in agent $(l,T_{lk})$ can be represented by
\begin{equation}\label{Sc_jian}
\widehat{c}_{(l,T_{lk})}^i(\widehat{A})=L_{\widehat{c}_{(l,T_{lk})}^i} \mathop{\ltimes}\limits_{l=1}^r  \mathop{\ltimes}\limits_{k=1}^{e_l} \widehat{A}_{{(l,T_{lk})}}
\end{equation}
and the algebraic form of agent $\left(l, T_{lk}\right) $'s payoff function can be denoted by
\begin{equation}\label{c_jian}
\widehat{C}_{(l,T_{lk})}(\widehat{A})=L_{\widehat{C}_{(l,T_{lk})}} \mathop{\ltimes}\limits_{l=1}^r  \mathop{\ltimes}\limits_{k=1}^{e_l} \widehat{A}_{{(l,T_{lk})}},
\end{equation}
where $L_{\widehat{c}_{(l,T_{lk})}^i} =P\ltimes \theta_{l,k,i}, i \in G_l$, and $L_{\widehat{C}_{(l,T_{lk})}}=P\ltimes \psi_{l,k}$.
\end{proposition}

\begin{proof} 
In $\widehat{\mathcal{G}}$, we arrange action profiles $ \widehat{A}=(\widehat{A}_{(1, T_{1 1})},\cdots,\widehat{A}_{(1, T_{1 e_1})},\cdots,\widehat{A}_{(r, T_{r 1})},\cdots,\widehat{A}_{(r, T_{r e_r})})$ in $\widehat{\mathcal{A}}$ in alphabetic order as
\begin{align}
&\widehat{A}^{1}=(A_{1 1},\cdots,A_{11},\cdots,A_{r1},\cdots,A_{r1}) , \nonumber  \\
&\widehat{A}^{2}=(A_{1 1},\cdots,A_{11},\cdots,A_{r1},\cdots,A_{r2}) ,\cdots , \nonumber \\
&\widehat{A}^{\widetilde{h}}=(A_{1 e_1},\cdots,A_{1e_1},\cdots,A_{re_r},\cdots,A_{re_r}). \nonumber
\end{align}
For every $l \in [1;r]$ and $k \in [1;e_l]$, we define the vector $L_{\hat{c}_{(l,T_{lk})}^i} \in \mathbb{R}_{1 \times \widetilde{h}}$ as
\begin{equation}\label{Clk1}
\operatorname{Col}_{\alpha}(L_{\hat{c}_{(l,T_{lk})}^i}):=\hat{c}_{(l,T_{lk})}^i(\widehat{A}^{\alpha}),\alpha \in [1;\widetilde{h}].
\end{equation}


By using the vector form of $\widehat{A}$, we can obtain that
\begin{equation}
\hat{c}_{(l,T_{lk})}^i(\widehat{A})=L_{\hat{c}_{(l,T_{lk})}^i} \mathop{\ltimes}\limits_{l=1}^r  \mathop{\ltimes}\limits_{k=1}^{e_l} \widehat{A}_{{(l,T_{lk})}}.
\end{equation}

From Theorem \ref{the3}, we have
\begin{equation}\label{Clk2}
\hat{c}_{(l,T_{lk})}^i(\widehat{A})
= P \ltimes \theta(l,k,i,\widehat{A}),
\end{equation}
where $\theta(l,k,i,\widehat{A})=\Xi(l,k) \ltimes \Phi^{T}(\widehat{A}) \ltimes ({\delta}_m^i)^{T} \ltimes {\mathbb{C}}^{T} \in \mathbb{R}_{e^m \times 1}$.

For every $l \in [1;r]$, $k \in [1;e_l]$ and $ i \in G_l$, we define the matrix $\theta_{l,k,i} \in \mathbb{R}_{e^m \times \widetilde{h}}$ as
\begin{equation}\label{thetalk}
\operatorname{Col}_{\alpha}(\theta_{l,k,i})=\theta(l,k,i,\widehat{A}^{\alpha}), \alpha \in [1;\widetilde{h}].
\end{equation}

By virtue of (\ref{Clk1}), (\ref{Clk2}) and (\ref{thetalk}), we get
$$
\operatorname{Col}_{\alpha}(L_{\hat{c}_{(l,T_{lk})}^i})=P\ltimes \operatorname{Col}_{\alpha}(\theta_{l,k,i}),\alpha \in [1;\widetilde{h}].
$$
Therefore, for each $l \in [1;r]$, $k \in [1;e_l]$ and $ i \in G_l$, we have
\begin{align}\label{LC1}
L_{\hat{c}_{(l,T_{lk})}^i}=P\ltimes \theta_{l,k,i}.
\end{align}

Similarly, we can easily get
\begin{equation}\label{LC}
L_{\widehat{C}_{(l,T_{lk})}}=P\ltimes \psi_{l,k},
\end{equation}
where 
\begin{align*}
    \operatorname{Col}_{\alpha}(\psi_{l,k})=&\psi(l,k,\widehat{A}^{\alpha})\\
    =&\frac{1}{m_l} \Xi(l,k) \ltimes \Phi^{T}(\widehat{A}^{\alpha}) \ltimes \sum_{i \in G_l} ({\delta}_m^i)^{T} \ltimes {\mathbb{C}}^{T}
\end{align*}
and $\alpha \in [1;\widetilde{h}]$.
\end{proof} 

\subsection{(Strongly) potential equation of the multi-group ex-ante agent game\label{second}}

In the following two theorems, we give the (strongly) potential equation of a MBG's MEAG to judge whether the MEAG is (strongly) potential. And if it is (strongly) potential, its (strongly) potential function can be obtained by solving the (strongly) potential equation.

\begin{theorem}\label{potential}
The MEAG $\widehat{\mathcal{G}}=\{\widehat{G}, \widehat{\mathcal{A}}, \widehat{C}\}$ of a MBG $\mathcal{G}=( G,\mathcal{T},\mathcal{A},p,C)$ is potential if and only if $\widehat{\mathcal{G}}$'s potential equation
\begin{equation}\label{po}
\Lambda \xi=\Psi P^{T},
\end{equation}
has a solution, where
{\setlength{\arraycolsep}{3pt}
\renewcommand{\arraystretch}{0.1}
\[
\Lambda=\left[\begin{array}{ccccc}
I_{\widetilde{h}} &\Lambda_{(1,T_{11})}^T  &0 &\cdots &0\\
I_{\widetilde{h}} &0  &\Lambda_{(1,T_{12})}^T   &\cdots &0\\
\vdots  &\vdots  & \vdots & \ddots  &\vdots \\
I_{\widetilde{h}} &0  &0   &\cdots &\Lambda_{(r,T_{re_r})}^T\\
\end{array}\right],
\]
}
$\Psi=[\psi_{1,1} , \psi_{1,2} , \cdots , \psi_{r,e_r}]^T$, $\xi=[\xi_0 , \xi_{(1,T_{11})} , \xi_{(1,T_{12})} , \cdots ,\\ \xi_{(r,T_{re_r})}]^T$ is unknown, $\xi_0 \in \mathbb{R}_{1 \times \widetilde{h}}$, and $\xi_{(l,T_{lk})} \in \mathbb{R}_{1 \times \frac{\widetilde{h}}{g_l}}$. In addition, if the solution of potential equation (\ref{po}) exists, the structure vector of $\widehat{\mathcal{G}}$'s potential function $\widehat{F}$ is
\begin{equation}\label{L_F_jian}
L_{\widehat{F}}=\xi_0^T.
\end{equation}
\end{theorem}

\begin{proof}
$\widehat{\mathcal{G}}$ is potential if and only if there exists a function $\widehat{F}$ which satisfies
\begin{align}\label{F_jian1}
&\widehat{C}_{(l,T_{lk})}( \widehat{A} _{\left(l, T_{lk}\right)}, \widehat{A}_{-\left(l, T_{lk}\right)})-\widehat{C}_{(l,T_{lk})}( \widehat{A}^{\prime}_{\left(l, T_{lk}\right)}, \widehat{A}_{-\left(l, T_{lk}\right)}) \nonumber \\
 =&\widehat{F}( \widehat{A} _{\left(l, T_{lk}\right)}, \widehat{A}_{-\left(l, T_{lk}\right)})-\widehat{F}( \widehat{A}^{\prime}_{\left(l, T_{lk}\right)}, \widehat{A}_{-\left(l, T_{lk}\right)}),
\end{align}
for any agent $(l,T_{lk}) \in \widehat{G}$ and any $\widehat{A} _{\left(l, T_{lk}\right)}, \widehat{A}^{\prime}_{\left(l, T_{lk}\right)} \in \widehat{\mathcal{A}}_{(l,T_{lk})}$, $ \widehat{A}_{-\left(l, T_{lk}\right)} \in \widehat{\mathcal{A}}_{-(l,T_{lk})}$.

It's easy to see that (\ref{F_jian1}) is equivalent to the equation
\begin{align}\label{F_jian2}
&\widehat{C}_{(l,T_{lk})}( \widehat{A} _{\left(l, T_{lk}\right)}, \widehat{A}_{-\left(l, T_{lk}\right)})-\widehat{F}( \widehat{A} _{\left(l, T_{lk}\right)}, \widehat{A}_{-\left(l, T_{lk}\right)})\nonumber \\
 =&\widehat{C}_{(l,T_{lk})}( \widehat{A}^{\prime}_{\left(l, T_{lk}\right)}, \widehat{A}_{-\left(l, T_{lk}\right)}) -\widehat{F}( \widehat{A}^{\prime}_{\left(l, T_{lk}\right)}, \widehat{A}_{-\left(l, T_{lk}\right)}),
\end{align}
which reveals that $\widehat{C}_{(l,T_{lk})}( \widehat{A} _{\left(l, T_{lk}\right)}, \widehat{A}_{-\left(l, T_{lk}\right)})-\widehat{F}( \widehat{A} _{\left(l, T_{lk}\right)}, \widehat{A}_{-\left(l, T_{lk}\right)})$ is independent of $\widehat{A} _{\left(l, T_{lk}\right)}$. For every $l \in [1;r]$, $k \in [1;e_l]$, we define
\begin{align}
H_{(l,T_{lk})}(A _{\left(l, T_{lk}\right)}, \widehat{A}_{-\left(l, T_{lk}\right)}):=\widehat{C}_{(l,T_{lk})}(\widehat{A})-\widehat{F}( \widehat{A}),\nonumber
\end{align}
where $A _{\left(l, T_{lk}\right)}$ represents the missing variable $ \widehat{A} _{\left(l, T_{lk}\right)}$. Therefore we have
\begin{align}\label{H}
\widehat{C}_{(l,T_{lk})}(\widehat{A})=\widehat{F}( \widehat{A})+H_{(l,T_{lk})}(A _{\left(l, T_{lk}\right)}, \widehat{A}_{-\left(l, T_{lk}\right)}).
\end{align}

Then we get that $\widehat{\mathcal{G}}$ is potential if and only if there exist functions $H_{(l,T_{lk})}(A _{\left(l, T_{lk}\right)}, \widehat{A}_{-\left(l, T_{lk}\right)}), l \in [1;r],k \in [1;e_l]$ which are independent of $\widehat{A} _{\left(l, T_{lk}\right)}$ satisfying (\ref{H}).

The vector form of (\ref{H}) can be expressed as
\begin{align}\label{H2}
L_{\widehat{C}_{(l,T_{lk})}} &\mathop{\ltimes}\limits_{l=1}^r  \mathop{\ltimes}\limits_{k=1}^{e_l} \widehat{A}_{{(l,T_{lk})}}=L_{\widehat{F}} \mathop{\ltimes}\limits_{l=1}^r  \mathop{\ltimes}\limits_{k=1}^{e_l} \widehat{A}_{{(l,T_{lk})}}+L_{H_{(l,T_{lk})}}\nonumber \\
&\mathop{\ltimes}\limits_{i=1}^{l-1}  \mathop{\ltimes}\limits_{j=1}^{e_i} \widehat{A}_{{(i,T_{ij})}}
\mathop{\ltimes}\limits_{j\in [1;e_l]}^{j \neq k}\widehat{A}_{{(l,T_{lj})}}\mathop{\ltimes}\limits_{i=l+1}^r  \mathop{\ltimes}\limits_{j=1}^{e_i} \widehat{A}_{{(i,T_{ij})}}.
\end{align}
where $L_{\widehat{F}}$, $L_{H_{(l,T_{lk})}}$ denote the structure vectors of $\widehat{F}$ and $H_{(l,T_{lk})}$, respectively.

Define
$$
\Lambda_{(l,T_{lk})}:=\mathop{\otimes}\limits_{i=1}^r \mathop{\otimes}\limits_{j=1}^{e_i} \lambda_{ij},
$$
where $\lambda_{lk}=\left\{\begin{array}{ll}
\mathbf{1}_{h_i}^{T}, & i=l,j=k, \\
I_{h_i}, & otherwise. \\
\end{array}\right.$

Then we have
$$
\begin{aligned}
\Lambda_{(l,T_{lk})}\mathop{\ltimes}\limits_{i=1}^r  \mathop{\ltimes}\limits_{j=1}^{e_i} \widehat{A}_{{(i,T_{ij})}}=&\mathop{\ltimes}\limits_{i=1}^{l-1}  \mathop{\ltimes}\limits_{j=1}^{e_i} \widehat{A}_{{(i,T_{ij})}}
\mathop{\ltimes}\limits_{j\in [1;e_l]}^{j \neq k}\widehat{A}_{{(l,T_{lj})}}\\
&\mathop{\ltimes}\limits_{i=l+1}^r  \mathop{\ltimes}\limits_{j=1}^{e_i} \widehat{A}_{{(i,T_{ij})}}.
\end{aligned}
$$

So (\ref{H2}) can be rewritten as
\begin{align}
L_{\widehat{C}_{(l,T_{lk})}} \mathop{\ltimes}\limits_{l=1}^r  \mathop{\ltimes}\limits_{k=1}^{e_l} \widehat{A}_{{(l,T_{lk})}}=&L_{\widehat{F}} \mathop{\ltimes}\limits_{l=1}^r  \mathop{\ltimes}\limits_{k=1}^{e_l} \widehat{A}_{{(l,T_{lk})}}+L_{H_{(l,T_{lk})}}\ltimes \nonumber \\
&\Lambda_{(l,T_{lk})}\mathop{\ltimes}\limits_{i=1}^r  \mathop{\ltimes}\limits_{j=1}^{e_i} \widehat{A}_{{(i,T_{ij})}}.
\end{align}

Because $\mathop{\ltimes}\limits_{l=1}^r  \mathop{\ltimes}\limits_{k=1}^{e_l} \widehat{A}_{{(l,T_{lk})}}$ are arbitrary, we get
\begin{align}
L_{\widehat{C}_{(l,T_{lk})}} =L_{\widehat{F}} +L_{H_{(l,T_{lk})}}\ltimes \Lambda_{(l,T_{lk})}.
\end{align}

By virtue of (\ref{LC}), we have
\begin{align}\label{PL}
P\ltimes \psi_{l,k} =L_{\widehat{F}} +L_{H_{(l,T_{lk})}}
\ltimes \Lambda_{(l,T_{lk})}.
\end{align}

By taking transpose of (\ref{PL}), we can obtain that for every $l \in [1;r]$, $k \in [1;e_l]$,
\begin{align}\label{theta}
 \psi^{T}_{l,k} \ltimes P^{T} =L^{T}_{\widehat{F}} +\Lambda^{T}_{(l,T_{lk})}\ltimes L^{T}_{H_{(l,T_{lk})}}.
\end{align}

Denote $\xi_0=L^{T}_{\widehat{F}}$, $\xi_{(l,T_{lk})}=L^{T}_{H_{(l,T_{lk})}}$, then (\ref{theta}) can be represented by (\ref{po}).
\end{proof} 

\begin{theorem}\label{the4}
The MEAG $\widehat{\mathcal{G}}=\{\widehat{G}, \widehat{\mathcal{A}}, \widehat{C}\}$ of a MBG $\mathcal{G}=( G,\mathcal{T},\mathcal{A},p,C)$ is strongly potential if and only if $\widehat{\mathcal{G}}$'s strongly potential equation
\begin{equation}\label{po1}
\Lambda \xi=\Theta P^{T},
\end{equation}
has a solution, where
{\setlength{\arraycolsep}{3pt}
\renewcommand{\arraystretch}{0.1}
$$
\Lambda=\left[\begin{array}{ccccc}
I_{\widetilde{h}} &\Lambda_{(1,T_{11},1)}^T  &0 &\cdots &0\\
I_{\widetilde{h}} &0  &\Lambda_{(1,T_{11},2)}^T   &\cdots &0\\
\vdots  &\vdots  & \vdots & \ddots  &\vdots \\
I_{\widetilde{h}} &0  &0   &\cdots &\Lambda_{(r,T_{re_r},\overline{m}_{r})}^T\\
\end{array}\right],
$$
}
\linespread{0.2} \selectfont
{\setlength{\arraycolsep}{4pt}
$$\xi=[\begin{array}{ccccc}\xi_0 & \xi_{(1,T_{11},1)} & \xi_{(1,T_{11},2)} & \cdots & \xi_{(r,T_{re_r},\overline{m}_{r})}\end{array}]^T,$$
}
{\setlength{\arraycolsep}{4.2pt}
$$\Theta=[\begin{array}{cccc}\theta_{1,1,1} & \theta_{1,1,2} & \cdots & \theta_{r,e_r,\overline{m}_{r}}\end{array}]^T,$$
}
$\xi$ is unknown, $\xi_0 \in \mathbb{R}_{1 \times \widetilde{h}}$, and $\xi_{(l,T_{lk},i)} \in \mathbb{R}_{1 \times \frac{\widetilde{h}}{g_l}}$. In addition, if the solution of strongly potential equation (\ref{po1}) exists, the structure vector of $\widehat{\mathcal{G}}$'s strongly potential function $\widehat{F}$ is
\begin{equation}\label{L_F_jian1}
L_{\widehat{F}}=\xi_0^T.
\end{equation}
\end{theorem}

\begin{proof}
$\widehat{\mathcal{G}}$ is strongly potential if and only if there exists a function $\widehat{F}$ which satisfies
\begin{align}\label{F_jian11}
&\hat{c}_{(l,T_{lk})}^i( \widehat{A} _{\left(l, T_{lk}\right)}, \widehat{A}_{-\left(l, T_{lk}\right)})-\hat{c}_{(l,T_{lk})}^i( \widehat{A}^{\prime}_{\left(l, T_{lk}\right)}, \widehat{A}_{-\left(l, T_{lk}\right)}) \nonumber \\
 =&\widehat{F}( \widehat{A} _{\left(l, T_{lk}\right)}, \widehat{A}_{-\left(l, T_{lk}\right)})-\widehat{F}( \widehat{A}^{\prime}_{\left(l, T_{lk}\right)}, \widehat{A}_{-\left(l, T_{lk}\right)}),
\end{align}
for any agent $(l,T_{lk}) \in \widehat{G}$ and any $i \in G_l$, $\widehat{A} _{\left(l, T_{lk}\right)}, \widehat{A}^{\prime}_{\left(l, T_{lk}\right)} \in \widehat{\mathcal{A}}_{(l,T_{lk})}$, $ \widehat{A}_{-\left(l, T_{lk}\right)} \in \widehat{\mathcal{A}}_{-(l,T_{lk})}$.

It's easy to see that (\ref{F_jian11}) is equivalent to the equation
\begin{align}\label{F_jian2}
&\hat{c}_{(l,T_{lk})}^i( \widehat{A} _{\left(l, T_{lk}\right)}, \widehat{A}_{-\left(l, T_{lk}\right)})-\widehat{F}( \widehat{A} _{\left(l, T_{lk}\right)}, \widehat{A}_{-\left(l, T_{lk}\right)})\nonumber \\
 =&\hat{c}_{(l,T_{lk})}^i( \widehat{A}^{\prime}_{\left(l, T_{lk}\right)}, \widehat{A}_{-\left(l, T_{lk}\right)}) -\widehat{F}( \widehat{A}^{\prime}_{\left(l, T_{lk}\right)}, \widehat{A}_{-\left(l, T_{lk}\right)}),
\end{align}
which reveals that $\hat{c}_{(l,T_{lk})}^i( \widehat{A} _{\left(l, T_{lk}\right)}, \widehat{A}_{-\left(l, T_{lk}\right)})-\widehat{F}( \widehat{A} _{\left(l, T_{lk}\right)}, \widehat{A}_{-\left(l, T_{lk}\right)})$ is independent of $\widehat{A} _{\left(l, T_{lk}\right)}$. For every $l \in [1;r]$, $k \in [1;e_l]$, we define
\begin{align}
H_{(l,T_{lk})}^i(A _{\left(l, T_{lk}\right)}, \widehat{A}_{-\left(l, T_{lk}\right)}):=\hat{c}_{(l,T_{lk})}^i(\widehat{A})-\widehat{F}( \widehat{A}),\nonumber
\end{align}
where $A _{\left(l, T_{lk}\right)}$ represents the missing variable $ \widehat{A} _{\left(l, T_{lk}\right)}$. Therefore we have
\begin{align}\label{H1}
\hat{c}_{(l,T_{lk})}^i(\widehat{A})=\widehat{F}( \widehat{A})+H_{(l,T_{lk})}^i(A _{\left(l, T_{lk}\right)}, \widehat{A}_{-\left(l, T_{lk}\right)}).
\end{align}

Then we get that $\widehat{\mathcal{G}}$ is strongly potential if and only if there exist functions $H_{(l,T_{lk})}^i(A _{\left(l, T_{lk}\right)}, \widehat{A}_{-\left(l, T_{lk}\right)}), l \in [1;r],k \in [1;e_l]$, $i \in G_l$ which are independent of $\widehat{A} _{\left(l, T_{lk}\right)}$ satisfying (\ref{H1}).

The vector form of (\ref{H1}) can be expressed as
\begin{align}\label{H21}
L_{\hat{c}_{(l,T_{lk})}^i} &\mathop{\ltimes}\limits_{b=1}^r  \mathop{\ltimes}\limits_{d=1}^{e_b} \widehat{A}_{{(b,T_{bd})}}=L_{\widehat{F}} \mathop{\ltimes}\limits_{b=1}^r  \mathop{\ltimes}\limits_{d=1}^{e_b} \widehat{A}_{{(b,T_{bd})}}+L_{H_{(l,T_{lk})}^i}\nonumber \\
&\mathop{\ltimes}\limits_{b=1}^{l-1}  \mathop{\ltimes}\limits_{d=1}^{e_b} \widehat{A}_{{(b,T_{bd})}}
\mathop{\ltimes}\limits_{d\in [1;e_l]}^{d \neq k}\widehat{A}_{{(l,T_{ld})}}\mathop{\ltimes}\limits_{b=l+1}^r  \mathop{\ltimes}\limits_{d=1}^{e_b} \widehat{A}_{{(b,T_{bd})}}.
\end{align}
where $L_{\widehat{F}}$, $L_{H_{(l,T_{lk})}^i}$ denote the structure vectors of $\widehat{F}$ and $H_{(l,T_{lk})}^i$, respectively.

Define
$$
\Lambda_{(l,T_{lk},i)}:=\mathop{\otimes}\limits_{b=1}^r \mathop{\otimes}\limits_{d=1}^{e_b} \lambda_{bd},
$$
where $\lambda_{bd}=\left\{\begin{array}{ll}
\mathbf{1}_{g_b}^{T}, & b=l,d=k, \\
I_{g_b}, & otherwise. \\
\end{array}\right.$

Then we have
$$
\begin{aligned}
\Lambda_{(l,T_{lk},i)}\mathop{\ltimes}\limits_{b=1}^r  \mathop{\ltimes}\limits_{d=1}^{e_b} \widehat{A}_{{(b,T_{bd})}}=&\mathop{\ltimes}\limits_{b=1}^{l-1}  \mathop{\ltimes}\limits_{d=1}^{e_b} \widehat{A}_{{(b,T_{bd})}}
\mathop{\ltimes}\limits_{d\in [1;e_l]}^{d \neq k}\widehat{A}_{{(l,T_{ld})}}\\
&\mathop{\ltimes}\limits_{b=l+1}^r  \mathop{\ltimes}\limits_{d=1}^{e_b} \widehat{A}_{{(b,T_{bd})}}.
\end{aligned}
$$

So (\ref{H21}) can be rewritten as
\begin{align}
L_{\hat{c}_{(l,T_{lk})}^i} \mathop{\ltimes}\limits_{b=1}^r  \mathop{\ltimes}\limits_{d=1}^{e_b} \widehat{A}_{{(b,T_{bd})}}=&L_{\widehat{F}} \mathop{\ltimes}\limits_{b=1}^r  \mathop{\ltimes}\limits_{d=1}^{e_b} \widehat{A}_{{(b,T_{bd})}}+L_{H_{(l,T_{lk})}^i}\ltimes \nonumber \\
&\Lambda_{(l,T_{lk},i)}\mathop{\ltimes}\limits_{b=1}^r  \mathop{\ltimes}\limits_{d=1}^{e_b} \widehat{A}_{{(b,T_{bd})}}.
\end{align}

Because $\mathop{\ltimes}\limits_{b=1}^r  \mathop{\ltimes}\limits_{d=1}^{e_b} \widehat{A}_{{(b,T_{bd})}}$ are arbitrary, we get
\begin{align}
L_{\hat{c}_{(l,T_{lk})}^i} =L_{\widehat{F}} +L_{H_{(l,T_{lk})}^i}\ltimes \Lambda_{(l,T_{lk},i)}.
\end{align}

By virtue of (\ref{LC1}), we have
\begin{align}\label{PL1}
P\ltimes \theta_{l,k,i} =L_{\widehat{F}} +L_{H_{(l,T_{lk})}^i}
\ltimes \Lambda_{(l,T_{lk},i)}.
\end{align}

By taking transpose of (\ref{PL1}), we can obtain that for every $l \in [1;r]$, $k \in [1;e_l]$,$i \in G_l$,
\begin{align}\label{theta1}
 \theta^{T}_{l,k,i} \ltimes P^{T} =L^{T}_{\widehat{F}} +\Lambda^{T}_{(l,T_{lk},i)}\ltimes L^{T}_{H_{(l,T_{lk})}^i}.
\end{align}

Denote $\xi_0=L^{T}_{\widehat{F}}$, $\xi_{(l,T_{lk},i)}=L^{T}_{H_{(l,T_{lk})}^i}$, then (\ref{theta1}) can be represented by (\ref{po1}).
\end{proof} 

Based on the above, we can get a (strongly) potential function of the MEAG which is (strongly) potential. And the agent action \textit{profiles} maximize this (strongly) potential function are (strongly) Nash equilibria of the MEAG.

\subsection{Algorithms and examples for finding (strongly) MBNE of a MBG\label{third}}

According to \Cref{th:thm1}, we can easily obtain the (strongly) MBNE of the MBG based on the (strongly) NE of the MBG's MEAG obtained above. In summary, we give \Cref{alg1,alg2} for detailed steps to find (strongly) MBNE of a MBG.

\begin{algorithm}[!htb]
\caption{Finding MBNE of a MBG}
\label{alg1}
\begin{algorithmic}
\STATE Transform a MBG $\mathcal{G}$ into its corresponding MEAG;\\ 
\STATE Calculate the algebraic form of agent $(l,T_{lk})$'s payoff function by \Cref{c_jian} ;\\
\IF{$\widehat{\mathcal{G}}$'s potential equation (\ref{po}) has a solution}
\STATE Calculate the structure vevtor of $\widehat{\mathcal{G}}$'s potential function $\widehat{F}$ by \Cref{L_F_jian};\\
\STATE Find the agent action \textit{profiles} $\widehat{A}^*$ that maximize the potential function $\widehat{F}$;\\
\STATE Calculate the MBG's MBNE by $\Gamma^{-1}(\widehat{A}^*)$;\\ 
\ELSE
\STATE The MEAG is not potential;\\
\ENDIF
\end{algorithmic}
\end{algorithm}

\begin{algorithm}[!htb]
\caption{Finding strongly MBNE of a MBG}
\label{alg2}
\begin{algorithmic}
\STATE Transform the MBG $\mathcal{G}$ into its corresponding MEAG;\\
\STATE Calculate the algebraic form of player $i$'s payoff function in agent $(l,T_{lk})$ by \Cref{Sc_jian};\\
\IF{$\widehat{\mathcal{G}}$'s strongly potential equation (\ref{po1}) has a solution}
\STATE Calculate the structure vevtor of $\widehat{\mathcal{G}}$'s strongly potential function $\widehat{F}$ by \Cref{L_F_jian1};\\
\STATE Find the agent action \textit{profiles} $\widehat{A}^*$ that maximize the strongly potential function $\widehat{F}$;\\
\STATE Calculate the MBG's strongly MBNE by $\Gamma^{-1}(\widehat{A}^*)$;\\ 
\ELSE
\STATE The MEAG is not strongly potential;\\
\ENDIF
\end{algorithmic}
\end{algorithm}

In the following, a numerical example based on \Cref{e1,e2} is given to show how to find MBNE of a MBG by using the multi-group ex-ante agent transformation. Given that the steps to find MBNE and strongly MBNE are similar, we will no longer give an example of finding strongly MBNE.

\begin{example}
For the auction problem modeled as a multi-group Bayesian game in \Cref{e1,e2}, find all of its multi-group Bayesian Nash equilibria by using \cref{alg1}.
\end{example}

According to \Cref{alg1}, we will find all MBNE of the MBG in the following. 

Firstly, based on the data given in \Cref{e1,e2}, we compute the probability vector $P$ of the probability distribution $p$ and the structural matrix $\mathbb{C}_i$ of player $i$'s payoff function, which are defined in \Cref{pt,cita1}. Arrange the types in $\mathcal{T}$ in lexicographical order as $T^1=(T_{11},T_{21})$, $T^2=(T_{11},T_{22})$, $\cdots$, $T^7=(T_{12},T_{23})$, $T^8=(T_{12},T_{24})$. {\rm Then we can get the probability vector} $P=[p(T^1),p(T^2),\cdots,p(T^8)]=[0.125,\\0.05,0.03,0.125,0.2,0.02,0.25,0.2]$ {\rm of the prior probability distribution $p$ in this MBG.} {\rm Similarly, arrange the type-action pairs in} $\mathcal{T} \times \mathcal{A}$ {\rm in lexicographical order as} $TA^1=(T_{11},T_{21},A_{11},A_{21})$, $\cdots$, $TA^{64}=(T_{12},T_{24},A_{12},A_{24})$. 
Since the auction rule in this auction problem is the first-price sealed-bid auction \cite{Elyakime1994}, we can get structure matrices of all players' payoff functions:
\begin{align*}
&\mathbb{C}_1=[c_1(TA^1),\cdots,c_1(TA^{64})]=[0,0,\cdots,0,0], \\ 
&\mathbb{C}_2=[c_2(TA^1),\cdots,c_2(TA^{64})]=[38,38,\cdots,3,3], \\
&\mathbb{C}_3=[c_3(TA^1),\cdots,c_3(TA^{64})]=[0,0,\cdots,0,0].
\end{align*}

Secondly, we transform the MBG into its corresponding multi-group ex-ante agent game $\widehat{\mathcal{G}}=\{\widehat{G},\widehat{\mathcal{A}}, \widehat{C}\}$ and calculate\\ all Nash equilibria of $\widehat{\mathcal{G}}$. In $\widehat{\mathcal{G}}$, $\widehat{G}=\{(1,T_{11}),(1,T_{12}),(2,T_{21}\\),(2,T_{22}),(2,T_{23}),(2,T_{24})\}$, $\widehat{\mathcal{A}}_{(1,T_{11})}=\widehat{\mathcal{A}}_{(1,T_{12})}=\mathcal{A}_{G_1}$, $\widehat{\mathcal{A}}_{(2,T_{21})}= \widehat{\mathcal{A}}_{(2,T_{22})}=\widehat{\mathcal{A}}_{(1,T_{23})}=\widehat{\mathcal{A}}_{(1,T_{24})}=\mathcal{A}_{G_2}$. Based on \Cref{pro2,pro3} and \Cref{the3}, we can obtain the algebraic form of each agent $(l,T_{lk})$'s payoff function, based on which the potential equation of $\widehat{\mathcal{G}}$ can be obtained by using \Cref{potential}. By solving this potential equation, we can know that $\widehat{\mathcal{G}}$ is potential and one of its potential functions $\widehat{F}$ which is showed in \Cref{fig2} can be obtained. Then we can get that $\max \widehat{F}(\widehat{A}) = \widehat{F}(\widehat{A}^{1})= \widehat{F}(\widehat{A}^{257})=\widehat{F}(\widehat{A}^{513})=\widehat{F}(\widehat{A}^{769})=\\6.855$, so $\widehat{A}^{1}=(A_{11},A_{11},A_{21},A_{21},A_{21},A_{21})$, $\widehat{A}^{257}=(A_{11}\\,A_{12},A_{21},A_{21},A_{21},A_{21})$, $\widehat{A}^{513}=(A_{12},A_{11},A_{21},A_{21},A_{21},\\A_{21})$ and $\widehat{A}^{769}=(A_{12},A_{12},A_{21},A_{21},A_{21},A_{21})$ are Nash equilibria of $\widehat{\mathcal{G}}$ which have been marked in \Cref{fig2}.

Finally, we calculate all MBNE of the MBG. According to \Cref{th:thm1}, we can know that all MBNE of the MBG are $s^{1}=\Gamma^{-1}(\widehat{A}^{1})=\left(s^{1}_1,  s^{1}_2\right)$, $s^{257}=\Gamma^{-1}(\widehat{A}^{257})=\left(s^{257}_1,  s^{257}_2\right)$, $s^{513}=\Gamma^{-1}(\widehat{A}^{513})=\left(s^{513}_1,  s^{513}_2\right)$, $s^{769}=\Gamma^{-1}(\widehat{A}^{769})=\left(s^{769}_1,  s^{769}_2\right)$.
So if group $1$'s type is $T_{11}$ or $T_{12}$, group $1$ should choose $A_{11}$ or $A_{12}$; if group $2$'s type is $T_{21}$, $T_{22}$, $T_{23}$, or $T_{24}$, group $2$ should choose action $A_{21}$. Only in this way, no group wants to adjust its strategy unilaterally if the other's strategy remain unchanged.

\begin{figure}
\centering
\includegraphics[width=3in]{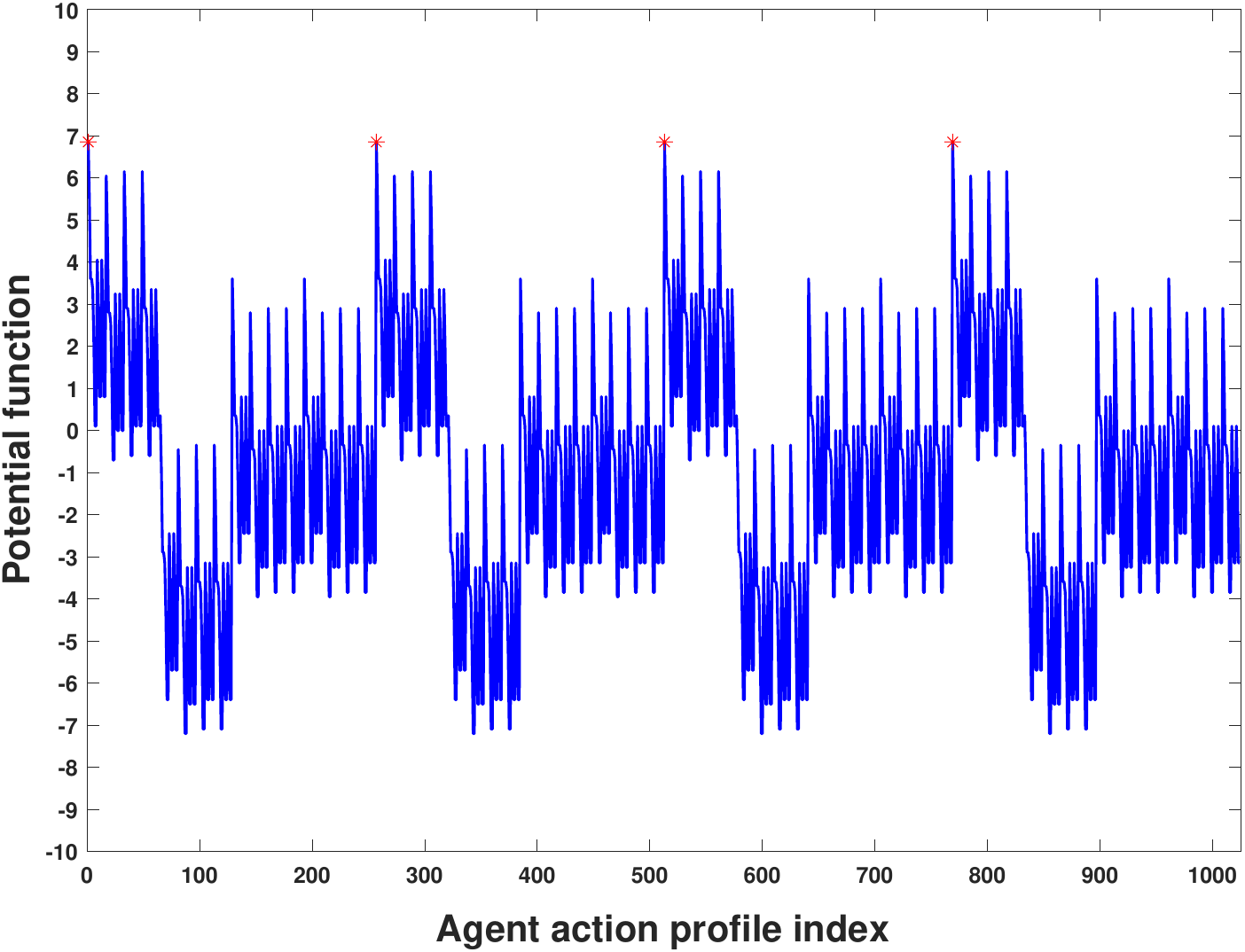}
\caption{A potential function and all MBNE of the MBG.}
\label{fig2}
\end{figure}

\begin{remark}
All examples in this paper were conducted using MATLAB R2021a. The software was operated on a Windows 10 platform with 16GB RAM and an Intel Core i7-10750H CPU, leveraging MATLAB's matrix-based computational framework, STP toolbox and other built-in toolboxes.
\end{remark}

\section{Conclusion}
\label{sec:conclusion}
In this paper, we have proposed the model of MBGs to investigate the group behavior among players in Bayesian games and have presented two types of definitions for Nash equilibria and potentiality to state different situations of group behavior. Then we have provided a transformation to convert any MBG into a MEAG for finding (strongly) MBNE of the MBG, and have offered the corresponding algorithms. Last, examples have been given to facilitate readers' understanding and verify the validity of this paper's results.

In the future, we will further attempt to study the dynamic behavior of each player in MBGs.

\bibliographystyle{IEEEtran}
\bibliography{bare_jrnl_new_sample4}
\end{document}